\documentclass{article}

\usepackage[colorinlistoftodos,bordercolor=orange,backgroundcolor=orange!20,line
color=orange,textsize=scriptsize]{todonotes}
\catcode`\@=11
\catcode`\@=12
\usepackage{subcaption}

\usepackage{arxiv}
\usepackage{wrapfig}

\usepackage{amsmath, amsthm, amssymb, epsfig, sectsty, graphicx, float, url}
\usepackage{bm}
\usepackage{algorithm, algorithmicx}
\usepackage[noend]{algpseudocode}

\usepackage[utf8]{inputenc} 
\usepackage[T1]{fontenc}    
\usepackage{hyperref}       
\usepackage{url}            
\usepackage{booktabs}       
\usepackage{amsfonts}       
\usepackage{nicefrac}       
\usepackage{microtype}      
\usepackage[capitalize,noabbrev]{cleveref}
\usepackage{caption}
\usepackage{subcaption}
\usepackage{inconsolata}
\usepackage{pifont}
\newcommand{\xmark}{\ding{55}}%

\usepackage{titlesec}
\titlespacing*{\paragraph} {0pt}{0.35ex plus 0.3ex minus .2ex}{1em}

\newcommand{\R}{\mathbb{R}}

\newcommand{\Z}{\mathcal{Z}}
\newcommand{\A}{\mathcal{A}}

\newcommand{\T}{\mathcal{T}}

\newcommand{\N}{\mathbb{N}}

\newcommand{\cX}{\mathcal{X}}
\newcommand{\cY}{\mathcal{Y}}
\newcommand{\cJ}{\mathcal{J}}

\newcommand{\J}{\mathcal{J}}

\newcommand{\C}{\mathcal{C}}

\newcommand{\cbap}{{\normalfont \texttt{CBA\textsuperscript{+}}}}
\newcommand{\tbp}{{\normalfont \texttt{TB\textsuperscript{+}}}}
\newcommand{\ptbp}{{\normalfont \texttt{PTB\textsuperscript{+}}}}
\newcommand{\smoothptbp}{{\normalfont \texttt{Smooth PTB\textsuperscript{+}}}}

\newcommand{\rmp}{{\normalfont \texttt{RM\textsuperscript{+}}}}
\newcommand{\prmp}{{\normalfont \texttt{PRM\textsuperscript{+}}}}
\newcommand{\cfrp}{{\normalfont \texttt{CFR\textsuperscript{+}}}}
\newcommand{\pcfrp}{{\normalfont \texttt{PCFR\textsuperscript{+}}}}

\newcommand{\adam}{{\normalfont \texttt {Adam}}}
\newcommand{\adagrad}{{\normalfont \texttt {AdaGrad}}}
\newcommand{\adamtbp}{{\normalfont \texttt {AdamTB\textsuperscript{+}}}}
\newcommand{\adatbp}{{\normalfont \texttt {AdaGradTB\textsuperscript{+}}}}
\newcommand{\omd}{{\normalfont \texttt{OMD}}}
\newcommand{\scpomd}{{\normalfont \texttt{SC-POMD}}}
\newcommand{\pomd}{{\normalfont \texttt{POMD}}}
\newcommand{\egt}{{\normalfont \texttt{EGT}}}

\newcommand{\Regmin}{{\sf Regmin}}

\newcommand{\cone}{{\sf cone}}
\newcommand{\diag}{{\sf diag}}

\newcommand{\tr}{^{\top}}
\newcommand{\Reg}{{\sf Reg}}
\newcommand{\Cgeq}{\mathcal{C}_{\geq}}

\newcommand{\xnot}{x_{\varnothing}}

\theoremstyle{plain}
\newtheorem{theorem}{Theorem}[section]
\newtheorem{lemma}[theorem]{Lemma}

\newtheorem{corollary}[theorem]{Corollary}

\newtheorem{proposition}[theorem]{Proposition}

\newtheorem{definition}[theorem]{Definition}
\newtheorem{remark}[theorem]{Remark}

\numberwithin{equation}{section}
\numberwithin{theorem}{section}

\title{Extensive-Form Game Solving via Blackwell Approachability on Treeplexes}
\author{%
Darshan Chakrabarti \\
IEOR Department\\
Columbia University \\
   \texttt{dc3595@columbia.edu} \\
\And
   Julien Grand-Cl{\'e}ment \\
ISOM Department\\
HEC Paris \\
   \texttt{grand-clement@hec.fr} \\
   \And
  Christian Kroer\\
  IEOR Department\\
Columbia University\\
  \texttt{christian.kroer@columbia.edu} \\
}

\begin{document}

\maketitle
\vspace{1cm}

\begin{abstract}
In this paper, we introduce the first algorithmic framework for Blackwell approachability on the sequence-form polytope, the class of convex polytopes capturing the strategies of players in extensive-form games (EFGs).
This leads to a new class of regret-minimization algorithms that are stepsize-invariant, in the same sense as the regret matching and regret matching$^+$ algorithms for the simplex.
Our modular framework can be combined with any existing regret minimizer over cones to compute a Nash equilibrium in two-player zero-sum EFGs with perfect recall, through the self-play framework. Leveraging predictive online mirror descent, we introduce {\em Predictive Treeplex Blackwell$^+$} (\ptbp), and show a $O(1/\sqrt{T})$ convergence rate to Nash equilibrium in self-play. We then show how to stabilize \ptbp{} with a stepsize, resulting in an algorithm with a state-of-the-art $O(1/T)$ convergence rate. 
We provide an extensive set of experiments to compare our framework with several algorithmic benchmarks, including \cfrp{} and its predictive variant, and we highlight interesting connections between practical performance and the stepsize-dependence or stepsize-invariance properties of classical algorithms.
\end{abstract}

\section{Introduction}
In this paper, we focus on solving {\em Extensive-Form Games} (EFGs), a widely used framework for modeling repeated games with perfect or imperfect information. Finding a Nash equilibrium of a two-player zero-sum EFG can be cast as solving a saddle-point problem of the form
\begin{equation}\label{eq:EFG-BSSP}
    \min_{\bm{x} \in \cX} \max_{\bm{y} \in \cY} \; \langle \bm{x},\bm{M}\bm{y}\rangle
\end{equation}
where the sets $\cX,\cY$  are two {\em sequence-form polytopes} (also referred to as {\em treeplexes}) representing the  strategies $\bm{x},\bm{y}$ of each player, and $\bm{M}$ is a payoff matrix. EFGs have been successfully used to obtain superhuman performances in several recent poker AI breakthroughs~\citep{tammelin2015solving,brown2018superhuman,brown2019superhuman}.

Many algorithms have been developed based on \eqref{eq:EFG-BSSP}. Since  $\cX$ and $\cY$ are polytopes, \eqref{eq:EFG-BSSP} can be formulated as a linear program~\citep{stengel1996efficient}. 
However, because $\cX$ and $\cY$ themselves have very large dimensions in realistic applications, {\em first-order methods} (FOMs) and {\em regret minimization} approaches are preferred for large-scale game solving. 
FOMs such as the Excessive Gap Technique (\egt,~\citet{nesterov2005excessive}) and Mirror Prox~\citep{nemirovski2004prox} instantiated for EFGs~\citep{hoda2010smoothing,kroer2018faster}  converge to a Nash equilibrium at a rate of $O(1/T)$, where $T$ is the number of iterations. Regret minimization techniques rely on a folk theorem relating the regrets of the players and the duality gap of the average iterates~\citep{freund1999adaptive}. For instance, predictive online mirror descent with with the treeplexes $\cX$ and $\cY$ as decision sets achieves a $O(1/T)$ convergence rate~\cite{farina2019optimistic}. 

{\em Counterfactual regret minimization} (CFR)~\cite{zinkevich2007regret} is a regret minimizer for the treeplex that runs regret minimizers {\em locally}, i.e. directly at the level of the information sets of each player. In particular, {\em \cfrp}, used in virtually all poker AI milestones~\cite{tammelin2015solving,moravvcik2017deepstack,brown2019superhuman}, instantiates the CFR framework with a regret minimizer called {\em Regret Matching}$^+$ (\rmp)~\cite{tammelin2015solving} and guarantees a $O(1/\sqrt{T})$ convergence rate.
The strong empirical performance of \cfrp{} remains mostly unexplained, since this algorithm does not achieve the fastest theoretical $O(1/T)$ convergence rate.
Interestingly, there is a stark contrast between the role of stepsizes in \cfrp{} versus in other algorithms. \cfrp{} may use different stepsizes across different infosets, and the iterates of \cfrp{} do not depend on the values of these stepsizes. We identify this property as {\em infoset stepsize invariance}. In contrast, the convergence properties of FOMs depend on the choice of a single stepsize used across the entire treeplex, which may be hard to tune in practice. 

\rmp{} arises as an instantiation of {\em Blackwell approachability}~\cite{blackwell1956analog} for the simplex. Blackwell approachability is a versatile framework with connections to online learning~\cite{abernethy2011blackwell} and applications in stochastic games~\cite{milman2006approachable}, calibration~\cite{perchet2010approachability} and revenue management~\cite{niazadeh2020online}.
Empirically, combining CFR with other local regret minimizers than \rmp{}, e.g. Online Mirror Descent (\omd), does not lead to the strong practical performance as \cfrp~\cite{brown2017dynamic}. This suggests that using a regret minimizer based on Blackwell approachability (\rmp) is central to the success of \cfrp.

Our goal is to develop Blackwell approachability-based algorithms for treeplexes. {\bf Our contributions} are as follows.

    \textbf{Treeplex Blackwell approachability.} We introduce the first Blackwell approachability-based regret minimizer for treeplexes. Using the self-play framework, we correspondingly get the first framework for solving two-player zero-sum EFGs via Blackwell approachability on treeplexes. 
    Blackwell approachability enables an equivalence between regret minimization over the treeplex $\T$ and regret minimization over its conic hull $\cone(\T)$, and any existing regret minimizer for $\cone(\T)$ yields a new algorithm for solving EFGs. A crucial advantage of using Blackwell approachability on the treeplex, rather than regret minimization directly on the treeplex, is that it leads to a variety of interesting stepsize properties (e.g. stepsize invariance), which are not achieved by regret minimizers such as \omd{} on the treeplex.
    We instantiate our framework with several regret minimizers, leading to different desirable properties. 
    
    
    \ptbp{} ({\em Predictive Treeplex Blackwell$^+$}, Algorithm \ref{alg:predictive-blackwell-treeplex-two-prox-calls}) combines our framework with predictive online mirror descent over $\cone(\T)$ and achieves a $O(1/\sqrt{T})$ convergence rate.
    As an instantiation of Blackwell approachability over a cone, \ptbp{} is {\em treeplex stepsize invariant}, i.e., its iterates do not change if we rescale all stepsizes by a positive constant. This is a desirable property for practical use, although it is a weaker property than the {\em infoset} stepsize invariance of \cfrp{}.

    \smoothptbp{} (Algorithm \ref{alg:pred-stable-blackwell-treeplex}) is a variant of \ptbp{} ensuring that successive iterates vary smoothly. This stability is a central element of the fastest algorithms for EFGs and we show that \smoothptbp{} has a $O(1/T)$ convergence rate. \smoothptbp{} is the first EFG-solving algorithm based on Blackwell approachability achieving the state-of-the-art theoretical convergence rate, answering an important open question. 
    
    \adatbp{}(Algorithm \ref{alg:adagrad-stable-blackwell-treeplex}) combines our framework with \adagrad~\cite{duchi2011adaptive} as a regret minimizer over $\cone(\T)$. \adatbp{} adaptively learns different stepsizes for every dimension of the treeplexes and guarantees a $O(1/\sqrt{T})$ convergence rate. For completeness, we consider \adamtbp, an adaptive instantiation of our framework based on \adam~\cite{kingma2014adam}, a widely used algorithm lacking regret guarantees.
    
    The convergence properties of our algorithms compared to existing methods can be found in Table \ref{tab:efg solving}.
    
    \textbf{Numerical experiments.} We provide two comprehensive sets of numerical experiments over benchmark EFGs. 
    
    We first compare the performance of all the algorithms introduced in our paper (Figure \ref{fig:tb_algo_comparison}) and find that \ptbp{} performs the best. This highlights the advantage of {\em treeplex} stepsize invariant algorithms (\ptbp) over stepsize-dependent algorithms, even ones achieving faster theoretical convergence rate (\smoothptbp), and over adaptive algorithms learning decreasing stepsizes (\adatbp). \adamtbp{} diverge on some instances, since \adam{} is not a regret minimizer.

    We then compare our best method (\ptbp) with several existing algorithms for EFGs: \cfrp, its predictive variants (\pcfrp), and predictive \omd{} (\pomd) (Figure \ref{fig:all_algo_avg_comparison}). We find that \pcfrp{} outperforms all other algorithms in terms of average-iterate performance. This suggests that {\em infoset} stepsize invariance is an important property, moreso than the treeplex stepsize invariance of \ptbp. Because of the CFR decomposition, \pcfrp{} can use different stepsizes at different infosets, where the values of the variables may be of very different magnitudes (typically, smaller for infosets appearing deeper in the treeplex), and \pcfrp{} does not require tuning these different stepsizes, which may be impossible for large instances. As part of our main contributions, we identify and distinguish the infoset stepsize invariance and treeplex stepsize invariance properties; based on our empirical experiments, we posit that part of the puzzle behind the strong empirical performances of \cfrp{} and \pcfrp{} can be explained by the infoset stepsize invariance property. We also compare the last-iterate performances, where no algorithms appear to consistently outperform the others. We leave studying the last-iterate convergence as an open question.

\begin{table}[htb]
\small
\centering
\begin{tabular}{lcc}
\toprule 
\bf Algorithms & \bf Convergence rate & \bf Stepsize invariance  \\
\midrule 
 \cfrp~\cite{tammelin2015solving} & $1/\sqrt{T}$ & $\checkmark\checkmark$ \\
 \pcfrp~\cite{farina2021faster} & $1/\sqrt{T}$ & $\checkmark\checkmark$\\
 \egt~\cite{kroer2018solving} & $1/T$ & \xmark \\
 \pomd~\cite{farina2019optimistic} & $1/T$ & \xmark \\
 \ptbp{} (Algorithm \ref{alg:predictive-blackwell-treeplex-two-prox-calls})  & $1/\sqrt{T}$& $\checkmark$ \\
  \smoothptbp{} (Algorithm \ref{alg:pred-stable-blackwell-treeplex})  & $1/T$ & \xmark \\
  \adatbp{} (Algorithm \ref{alg:adagrad-stable-blackwell-treeplex})  & $1/\sqrt{T}$ & \xmark \\
  \adamtbp{} (Algorithm \ref{alg:adam-stable-blackwell-treeplex})  & $?$  & \xmark  \\
\bottomrule
    \end{tabular}
    \caption{Convergence rates to a Nash equilibrium of a two-player zero-sum EFG for several algorithms. $\checkmark \checkmark$ refers to {\em infoset} stepsize invariance and $\checkmark$ refers to {\em treeplex} stepsize invariance.
    }\label{tab:efg solving}
\end{table}

\section{Preliminaries on EFGs}
We first provide some background on EFGs and treeplexes.
\paragraph{Extensive-form games.} Two-player zero-sum extensive-form games (later referred to as {\em EFGs}) are represented by a game tree and a payoff matrix. Each node of the tree belongs either to the first player, to the second player, or to a {\em chance player}, modeling the random events that happen in the game, e.g., tossing a coin. The players are assigned payoffs at the terminal nodes only. 
Imperfect/private information is modeled using {\em information sets} (later referred to as {\em infosets}), which are subsets of nodes of the game tree. A player cannot distinguish between the nodes in a given infoset, and they must take the same action at all these nodes.
\paragraph{Treeplexes.} The strategy of a player can be described by a polytope called the {\em treeplex}, also known as the {\em sequence-form polytope}. The treeplex is constructed as follows. We index the infosets of a player by $\cJ = \{ 1,...,|\cJ|\}$. The set of actions available at infoset $j \in \cJ$ is written $\A_{j}$ with cardinality $|\A_{j}|=n_{j}$. We represent choosing action $a \in \A_{j}$ at infoset $j \in \cJ$ by a {\em sequence} $(j,a)$, and we denote by $\C_{ja}$ the set of next infosets reachable from $(j,a)$ (possibly empty if the game terminates). The parent $p_{j}$ of an infoset $j \in \cJ$ is the sequence leading to $j$; note that $p_{j}$ is unique assuming perfect recall. We assume that there is a single root denoted as $\varnothing$ and called the {\em empty sequence}. If the player does not take any action before reaching $j \in \cJ$, then by convention $p_{j} = \varnothing$. Under the perfect recall assumption, the set of infosets has a tree structure: $\C_{ja} \cap \C_{j'a'} = \emptyset$, for all pairs of sequences $(j,a)$ and $(j',a')$ such that $j \neq j', a \neq a'$. This tree is the treeplex and it represents the set of all admissible strategies for a given player. We denote by $n \in \N$ the total number of sequences $(j,a)$ with $j \in \cJ$ and $a \in \A_{j}$.
With these notations, the treeplex $\T$ of a given player can be written as 
\begin{equation}\label{eq:treeplex}
\T = \{ \bm{x} \in \R^{n+1}_{+} \; | \; \xnot = 1, \sum_{a \in \A_{j}} x_{ja} = x_{p_{j}}, \forall \; j \in \J\}
\end{equation}
where the first component $\xnot$ is related to the empty sequence $\varnothing$. We say that a player {\em makes an observation} to arrive at $j$, if $|\C_{p_j}| > 1$.
We define the depth $d$ of a treeplex to be the maximum number of actions and observations that can be made starting at the root until reaching a leaf infoset.
Computing a Nash equilibrium of EFGs can be formulated as solving \eqref{eq:EFG-BSSP} (under the perfect recall assumption), with $\cX \subset \R^{n_{1}+1}$ and $\cY \subset \R^{n_{2}+1}$ the treeplex of each player, $n_{1}$ and $n_{2}$ are the number of sequences of each player, and $\bm{M} \in \R^{(n_{1}+1) \times (n_{2}+1)}$ the payoff matrix such that for a pair of strategy $(\bm{x},\bm{y}) \in \cX \times \cY$, $\langle \bm{x},\bm{M}\bm{y}\rangle$ is the expected value that the second player receives from the first player.

\paragraph{Regret minimization and self-play framework.} A {\em regret minimizer} \Regmin{} over a decision set $\Z \subset \R^{d}$ is an algorithm such that, at every iteration, \Regmin{} chooses a decision $\bm{z}^t \in \Z$, a {\em loss vector} $\bm{\ell} \in \R^{d}$ is observed, and the scalar loss $\langle \bm{\ell}^t,\bm{x}^t \rangle$ is incurred. A regret minimizer ensures that the {\em regret} $\Reg^{T}=\max_{\hat{\bm{z}} \in \Z} \sum_{t=1}^{T} \langle \bm{\ell}^{t},\bm{z}^t - \hat{\bm{z}} \rangle$ grows at most as $O(\sqrt{T})$. As an example, 
{\em predictive online mirror descent} (\pomd, \citet{rakhlin2013online}) generates a sequence of decisions $\bm{z}_{1},...,\bm{z}_{T} \in \Z$ as follows:
\begin{equation}\label{eq:predictive OMD}
   \begin{aligned}
         \bm{z}_{t} & = \Pi_{\Z}\left(\hat{\bm{z}}_{t} - \eta \bm{m}_{t}\right) 
         \\\hat{\bm{z}}_{t+1} & = \Pi_{\Z}\left(\hat{\bm{z}}_{t} - \eta \bm{\ell}_{t}\right)
   \end{aligned}
\end{equation}
with $\bm{m}_{1},...,\bm{m}_{T} \in \R^{d}$ some predictions of the losses $\bm{\ell}_{1},...,\bm{\ell}_{T} \in \R^{d}$, and where we write the orthogonal projection of $\bm{y} \in \R^{d}$ onto $\Z$ as $\Pi_{\Z}\left(\bm{y}\right) := \arg \min_{\bm{z} \in \Z} \| \bm{z} - \bm{y}\|_{2}$.

The {\em self-play framework} leverages regret minimization to solve EFGs. The players compute two sequences of strategies $\bm{x}_{1},...,\bm{x}_{T}$ and $\bm{y}_{1},...,\bm{y}_{T}$ such that, at iteration $t \geq 1$, the first player observes its loss vector $\bm{A}\bm{y}_{t-1}$ and the second player observes its loss vector $-\bm{A}\tr\bm{x}_{t-1}$. Each player computes their current strategies $\bm{x}_t \in \cX$ and $\bm{y}_t \in \cY$ by using a regret minimizer. A well-known folk theorem states that the duality gap of the average of the iterates is bounded by the sum of the average regrets of the players.
\begin{proposition}[\citep{freund1999adaptive}]\label{prop:self-play framework}
Let $\bm{x}_{1},...,\bm{x}_{T} \in \cX$ and $\bm{y}_{1},...,\bm{y}_{T} \in \cY$ be computed in the self-play framework. Let
$ \left(\bar{\bm{x}}_{T},\bar{\bm{y}}_{T}\right) = \frac{1}{T} \sum_{t=1}^{T}\left(\bm{x}_{t},\bm{y}_{t}\right)$. Then, for $\Reg^{T}_{1}$ and $\Reg^{T}_{2}$ the regret of each player, 
\[ \max_{\hat{\bm{y}} \in \cY} \; \langle \bar{\bm{x}}_{T},\bm{M}\hat{\bm{y}}\rangle - \min_{\hat{\bm{x}} \in \cX} \; \langle \hat{\bm{x}},\bm{M}\bar{\bm{y}}_{T}\rangle = \frac{\Reg^{T}_{1}+\Reg^{T}_{2}}{T}.\]
\end{proposition}
We present more details on the self-play framework in Appendix \ref{app:self-play framework}.
Several regret minimizers exist for the treeplex, e.g. \pomd{}~\citep{farina2021better}. 

\paragraph{CFR and Regret Matching$^+$.} {\em Counterfactual Regret minimization} (CFR, \cite{zinkevich2007regret}) runs independent regret minimizers with counterfactual losses at each infoset of the treeplexes. This considerably simplifies the optimization problem, since the decision set at each infoset $j \in \cJ$ is the simplex over the set of next available actions $\Delta^{n_{j}}:= \{ \bm{x} \in \R_{+}^{n_{j}} \; | \; \sum_{i=1}^{n_{j}} x_{i} = 1\}$.
In the CFR framework, the regret of each player (over the treeplex) is bounded by the maximum of the local regrets incurred at each infoset. Therefore, CFR combined with any regret minimizer over the simplex converges to a Nash equilibrium at a rate of $O(1/\sqrt{T})$. We refer to Appendix \ref{app:cfr} for more details.
Combining CFR with a local regret minimizer called {\em Regret Matching}$^+$ (\rmp, \citet{tammelin2015solving}) along with alternation and linear averaging yields an algorithm called \cfrp, which has been observed to attain strong practical performance compared to theoretically-faster methods~\citep{kroer2018faster}.
\rmp{} can only be implemented on the simplex and not for other decision sets, and proceeds as follows: given a sequence of loss $\bm{\ell}_{1},...,\bm{\ell}_{T} \in \R^{d}$, \rmp{} maintains a sequence $\bm{R}_{1},...,\bm{R}_{T} \in \R^{d}$ such that $\bm{R}_{1} = \bm{0}$ and
\begin{align}\label{eq:rmp}
    \bm{x}_{t} & = \bm{R}_{t} / \| \bm{R}_{t}\|_{1}, \bm{R}_{t+1} = \Pi_{\R_{+}^{d}}\left(\bm{R}_{t} - \eta \bm{g}(\bm{x}_{t},\bm{\ell}_{t})\right)
\end{align}
with $\eta>0$ and, for $\bm{x},\bm{\ell} \in \R^{d}$,
\begin{align}\label{eq:definition g}
\bm{g}(\bm{x},\bm{\ell}) & := \bm{\ell} - \langle \bm{x},\bm{\ell}\rangle \bm{1}.
\end{align}
We use the convention $\bm{0}/0 := (1/d)\bm{1}$ with $\bm{1} := (1,...,1) \in \R^{d}$. 
\rmp{} is {\em stepsize invariant}:  $\bm{x}_{1},...,\bm{x}_{T}$ are independent of $\eta$, since we have $\bm{x}_{t} = \bm{R}_{t}/\| \bm{R}_{t}\|_{1}$ and $\eta$ only rescales the entire sequence $\bm{R}_{1},...,\bm{R}_{T}$. Since \cfrp{} runs \rmp{} at each infoset independently, \cfrp{} is {\em infoset stepsize invariant}: there may be different stepsizes across different infosets and the iterates of \cfrp{} do not depend on their values, a desirable property when solving large-scale EFGs where stepsize tuning may be difficult. 

\rmp{} can be interpreted as a special instantiation of {\em Blackwell approachability}~\citep{blackwell1956analog,abernethy2011blackwell}. In this interpretation of \rmp, the goal of the decision maker is to compute the sequence of strategies $\bm{x}_{1},...,\bm{x}_{T} \in \Delta^d$ to ensure that the auxiliary sequence $\bm{R}_{T}/T \in \R_{+}^{d}$ approaches the {\em target set} $\R_{-}^{d}$ as $T \rightarrow + \infty$. Since $\bm{R}_{t} \in \R_{+}^{d}$, this is equivalent to ensuring that $ \lim_{T \rightarrow + \infty} \bm{R}_{T}/T = \bm{0}$. The vector $\bm{g}(\bm{x},\bm{\ell})$ is interpreted as an instantaneous loss for the approachability instance.
As an instantiation of Blackwell approachability, at each iteration \rmp{} computes an orthogonal projection onto the {\em conic hull} of the decision set:
\begin{equation}\label{eq:cone simplex}
    \R_{+}^{d} = \cone(\Delta^d)
\end{equation}
with $\cone(\Z) := \{ \alpha \bm{x} \; | \; \bm{x} \in \Z, \alpha \geq 0\}$ for a set $\Z$.
The decision function $\bm{R} \mapsto \bm{R} / \| \bm{R} \|_{1}$ is based on 
\begin{equation}\label{eq:simplex subset hyperplane}
    \Delta^{d} \subset \{ \bm{x} \in \R^{d} \; | \; \langle \bm{x},\bm{1}\rangle = 1\}.
\end{equation}
Since for $\bm{R} \in \R_{+}^{d}, \langle \bm{R},\bm{1}\rangle = \| \bm{R}\|_{1}$, then $\bm{x}_{t} = \bm{R}_{t} / \| \bm{R}_{t}\|_{1}$ can be written $\bm{x}_{t} = \bm{R}_{t} / \langle \bm{R}_{t},\bm{1}\rangle$, with $\bm{1}$ a vector such that the decision set $\Delta^d$ satisfies \eqref{eq:simplex subset hyperplane}. This ensures that
\begin{equation}\label{eq:hyperplane forcing simplex}
    \langle \bm{R}_{t},\bm{g}(\bm{x}_{t},\bm{\ell})\rangle = 0, \forall \; \bm{\ell} \in \R^{d}.
\end{equation}
We provide an illustration of the dynamics of \rmp{} in Figure \ref{fig:rmp dynamics}.
Equation \eqref{eq:hyperplane forcing simplex} is known as a {\em hyperplane forcing condition} and is a key ingredient in any Blackwell approachability-based algorithm; it ensures that the vector $\bm{R}_{T}$ grows at most at a rate of $O(\sqrt{T})$ so that $\lim_{T \rightarrow + \infty} \bm{R}_{T}/T = \bm{0}$. We refer to \citet{perchet2010approachability,grand2023solving} for more details on Blackwell approachability.

\begin{figure}
\centering
\includegraphics[width=0.3\columnwidth]{"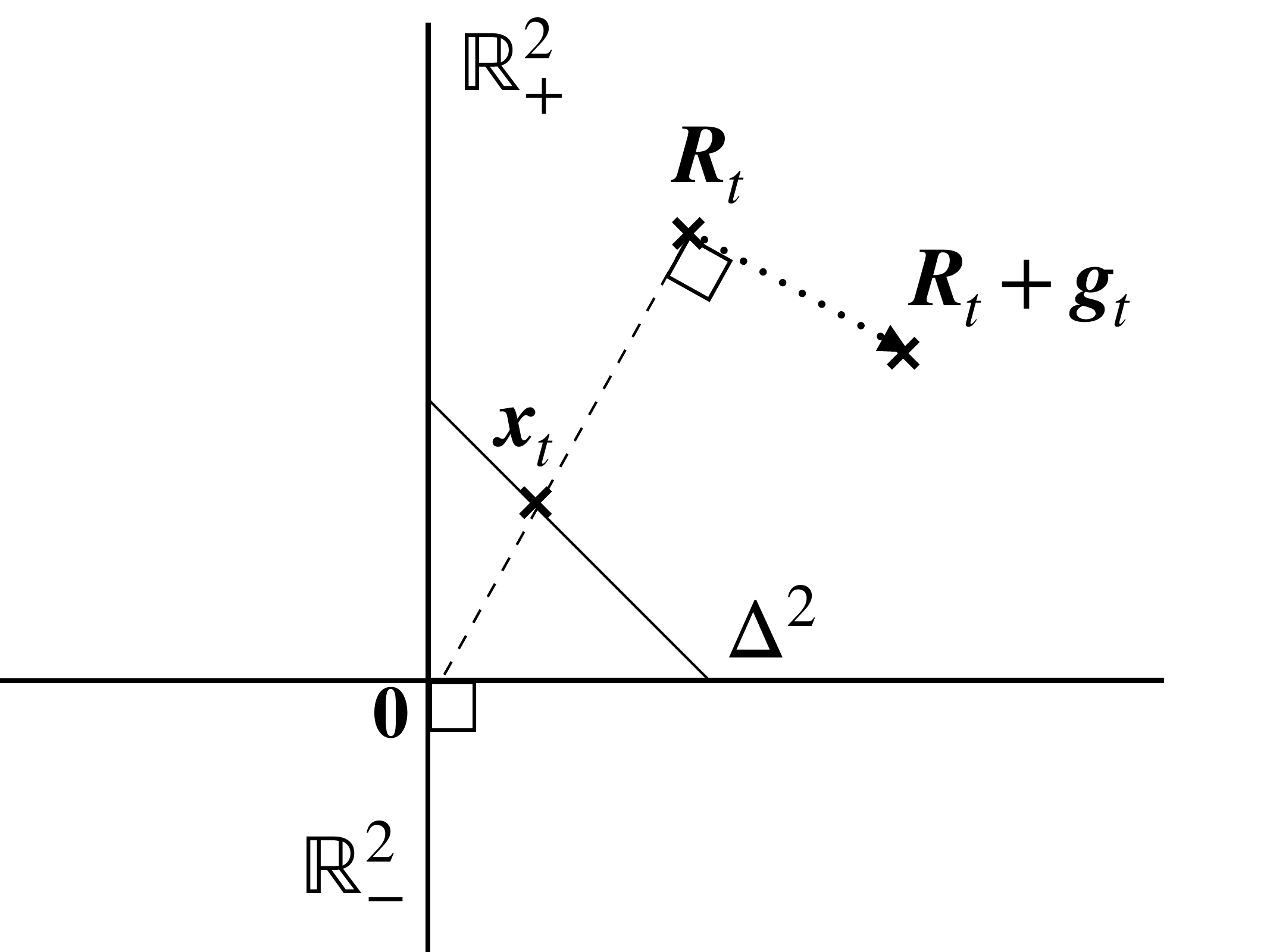"}
\caption{Dynamics of \rmp{} in $\R_{+}^{2}$. We write $\bm{g}_{t}=\bm{g}(\bm{x}_{t},\bm{\ell}_{t})$.}
\label{fig:rmp dynamics}
\end{figure}

\section{Blackwell Approachability on Treeplexes}
In this section we introduce a modular regret minimization framework for the treeplex based on Blackwell approachability. This framework can be used as a regret minimizer over $\T$ in the self-play framework (described in the previous section and in Appendix \ref{app:self-play framework}) to obtain an algorithm for solving EFGs.
Our algorithms are based on the fact that for $\T \subset \R^{n+1}$ a treeplex as defined in \eqref{eq:treeplex}, we have
\begin{equation}\label{eq:treeplex subspace}
    \T \subset \{ \bm{x} \in \R^{n+1} \; | \; \langle\bm{x},\bm{a}\rangle = 1\}
\end{equation}
for $\bm{a} = (1,\bm{0}) \in \R^{n+1}$ with $\bm{0} = (0,...,0) \in \R^{n}$. This property is analogous to \eqref{eq:simplex subset hyperplane} for the simplex. With this analogy in mind, we define $\C \subset \R^{n+1}$ and $\bm{f}(\bm{x},\bm{\ell}) \in \R^{n+1}$ as, for $\bm{x},\bm{\ell} \in \R^{n+1}$,
\begin{align}
    \C & :=\cone(\T) \label{eq:cone T} \\
\bm{f}(\bm{x},\bm{\ell})& :=\bm{\ell} - \langle  \bm{x},\bm{\ell}\rangle \bm{a}. \label{eq:definition-f-x-ell}
\end{align}
Equation \eqref{eq:cone T} is analogous to \eqref{eq:cone simplex} and Equation \eqref{eq:definition-f-x-ell} is analogous to \eqref{eq:definition g}. The cone $\C$ and the vector $\bm{f}(\bm{x},\bm{\ell})$ play a similar role for $\T$ as $\R_{+}^{d}$ and $\bm{g}(\bm{x},\bm{\ell})$ play for $\Delta^d$ in \rmp. 

Our Blackwell approachability-based framework is described in Algorithm \ref{alg:blackwell-approachability based regmin} and relies on running a regret minimizer \Regmin{} over $\C=\cone(\T)$ against the losses $\bm{f}(\bm{x}_{t},\bm{\ell}_{t})$ to obtain a regret minimizer over $\T$ against the losses $\bm{\ell}_{t}$, for $t \geq 1$.
\begin{algorithm}
      \caption{Blackwell approachability on the treeplex}
      \label{alg:blackwell-approachability based regmin}
      \begin{algorithmic}[1]
      \State {\bf Input}: A regret minimizer $\Regmin$ with decision set $\C$ 

      \State {\bf Initialization}: $\bm{R}_{1} = \bm{0} \in \R^{n+1}$
      
      \For{$t = 1, \dots, T$}
      \State $\bm{x}_{t} = \bm{R}_{t}/\langle\bm{R}_{t},\bm{a}\rangle$
      \State Observe the loss vector $\bm{\ell}_{t} \in \R^{n+1}$
      \State \Regmin{} observes $\bm{f}(\bm{x}_t,\bm{\ell}_t) \in \R^{n+1}$
      \State $\bm{R}_{t+1} = \Regmin\left(\cdot \right)$
      \EndFor
\end{algorithmic}
\end{algorithm}
We use the convention that $\bm{0}/0$ is the uniform strategy for the treeplex. Algorithm \ref{alg:blackwell-approachability based regmin} is the first Blackwell approachability-based algorithm operating on the entire treeplex (in contrast to \cfrp{} which relies on Blackwell approachability locally at the infosets level).
We first describe some important properties of
Algorithm \ref{alg:blackwell-approachability based regmin}:

{\em Feasibility of the current iterate.} Algorithm \ref{alg:blackwell-approachability based regmin} produces a feasible sequence of strategies, i.e., $\bm{x}_{t} \in \T, \forall \; t \geq 1$. Indeed, since \Regmin{} is a regret minimizer with $\C$ as the decision set, $\bm{R}_{t} \in \cone(\T)$, i.e., $\bm{R}_{t} = \alpha \bm{z}$ with $\alpha \in \R_{+}$ and $\bm{z} \in \T$. From \eqref{eq:treeplex subspace}, we have $\langle \bm{z},\bm{a}\rangle = 1$. Therefore,
$\bm{x}_{t} = \frac{\bm{R}_{t}}{\langle \bm{R}_{t},\bm{a}\rangle }= \frac{\alpha \bm{z}}{\alpha \langle \bm{z},\bm{a}\rangle} = \bm{z} \in \T.$ This is analogous to \rmp, where $\bm{x}_t$ is proportional to $\bm{R}_{t}$, see \eqref{eq:rmp} and Figure \ref{fig:rmp dynamics}.

{\em Hyperplane forcing.} For any $t \in \N$ we have
\begin{equation}\label{eq:halfspace-forcing}
\langle \bm{R}_{t},f(\bm{x}_{t},\bm{\ell}) \rangle =0, \forall \; \bm{\ell} \in \R^{n+1}.
\end{equation}
The hyperplane forcing equation~\eqref{eq:halfspace-forcing} is a crucial component of algorithms based on Blackwell approachability. It ensures that $\|\bm{R}_{t}\|_{2} = O(\sqrt{T})$. Equation~\eqref{eq:halfspace-forcing} is analogous to \eqref{eq:hyperplane forcing simplex} for \rmp{} and follows from $\bm{x}_{t} = \frac{\bm{R}_{t}}{ \langle \bm{R}_{t},\bm{a}\rangle}$, so that
\begin{align*}
    \langle \bm{R}_{t},\bm{f}(\bm{x}_{t},\bm{\ell}) \rangle & = \langle \bm{R}_{t},\bm{\ell}\rangle - \langle \bm{x}_{t},\bm{\ell}\rangle \langle \bm{R}_{t},\bm{a}\rangle \\
    & = \langle \bm{R}_{t},\bm{\ell}\rangle - \langle \frac{\bm{R}_{t}}{\langle \bm{R}_{t},\bm{a}\rangle},\bm{\ell}\rangle \langle \bm{R}_{t},\bm{a}\rangle \\
    & = \langle \bm{R}_{t},\bm{\ell}\rangle - \langle \bm{R}_{t},\bm{\ell}\rangle =0.
\end{align*}
{\em Regret minimization over $\T$.} Crucially, Algorithm \ref{alg:blackwell-approachability based regmin} always yields a regret minimizer over the treeplex $\T$, i.e., it ensures that the regret of $\bm{x}_{1},...,\bm{x}_{T} \in \T$ against any the sequence $\bm{\ell}_{1},...,\bm{\ell}_{T} \in \R^{n+1}$ is bounded by $O(\sqrt{T})$.
The proof is instructive and shows a central component to Blackwell approachability-based algorithms: minimizing regret over $\T$ can be achieved by minimizing regret over $\cone(\T)$.
\begin{proposition}\label{prop:blackwell-approachability adversarial-regret}
Let $\Regmin$ be a regret minimizer with $\C$ as the decision set.
Let $\bm{x}_{1},...,\bm{x}_{T} \in \T$ be computed by Algorithm \ref{alg:blackwell-approachability based regmin}. Then
$\max_{\hat{\bm{x}} \in \T} \sum_{t=1}^{T} \langle \bm{x}_{t}- \hat{\bm{x}},\bm{\ell}_{t}\rangle =O(\sqrt{T}).$
\end{proposition}
\begin{proof}
Let $\hat{\bm{x}} \in \T$ and let us write $\hat{\bm{R}} = \hat{\bm{x}}$. We have
\begin{align*}
     \sum_{t=1}^{T} \langle \bm{x}_{t}-\hat{\bm{x}},\bm{\ell}_{t} \rangle & = \sum_{t=1}^{T} \langle - \hat{\bm{x}},\bm{f}\left(\bm{x}_{t},\bm{\ell}_{t}\right) \rangle \\
    & =  \sum_{t=1}^{T} \langle- \hat{\bm{R}},\bm{f}\left(\bm{x}_{t},\bm{\ell}_{t}\right) \rangle\\
    & =  \sum_{t=1}^{T} \langle \bm{R}_{t}-\hat{\bm{R}},\bm{f}\left(\bm{x}_{t},\bm{\ell}_{t}\right) \rangle 
\end{align*}
where the first equality follows from the definition of $\bm{f}\left(\bm{x}_{t},\bm{\ell}_{t}\right)$ as in \eqref{eq:definition-f-x-ell} and $\langle \bm{z},\bm{a}\rangle =1$ for any $\bm{z} \in \T$, the second equality is because $\hat{\bm{x}} = \hat{\bm{R}}$, and the last equality follows from the hyperplane forcing condition \eqref{eq:halfspace-forcing}.
Now note that $\sum_{t=1}^{T} \langle \bm{R}_{t}-\hat{\bm{R}},\bm{f}\left(\bm{x}_{t},\bm{\ell}_{t}\right) \rangle$ is the regret of a regret minimizer \Regmin{} choosing $\bm{R}_{1},...,\bm{R}_{T}$ in the decision set $\C:= \cone(\T)$ against a sequence of loss $\bm{f}\left(\bm{x}_{1},\bm{\ell}_{1}\right),...,\bm{f}\left(\bm{x}_{T},\bm{\ell}_{T}\right)$ and a comparator $\hat{\bm{R}} \in \cone(\T)$. Therefore, $\sum_{t=1}^{T} \langle \bm{R}_{t}-\hat{\bm{R}},\bm{f}\left(\bm{x}_{t},\bm{\ell}_{t}\right) \rangle = O(\sqrt{T})$.
\end{proof}
\begin{remark}[Comparison with Lagrangian Hedging]
    Algorithm~\ref{alg:blackwell-approachability based regmin} is related to {\em Lagrangian Hedging}~\citep{gordon2007no,d2021optimisti}. Lagrangian Hedging builds upon Blackwell approachability with various potential functions to construct regret minimizers for {\em general} decision sets. As explained in the introduction, the main focus of our paper is on two-player zero-sum EFGs, i.e., on the case where the decision sets are treeplexes, and where we can obtain several additional interesting properties not studied in \citep{gordon2007no,d2021optimisti}, such as stepsize invariance, fast convergence rates, and efficient projection, as we detail in the next section.
    If one were to instantiate \cref{alg:blackwell-approachability based regmin} with the Follow-The-Regularized Leader algorithm, it would yield the regret minimizer for treeplexes studied in \citet{gordon2007no}, and our \cref{prop:practical implementation ptbp} yields an efficient projection oracle for the setup in \citet{gordon2007no}, which appealed to general convex optimization as an oracle.
\end{remark}

\section{Instantiations of Algorithm \ref{alg:blackwell-approachability based regmin}}
We can instantiate Algorithm \ref{alg:blackwell-approachability based regmin} with any regret minimizer over $\C$ to obtain various properties for the resulting algorithm, such as stepsize invariance, adaptive stepsizes, or achieving the state-of-the-art $O(1/T)$ convergence rate. We show next how to do so.
\paragraph{Predictive Treeplex Blackwell$^+$ (\ptbp).}
We first introduce {\em Predictive Treeplex Blackwell}$^+$ (Algorithm \ref{alg:predictive-blackwell-treeplex-two-prox-calls}), combining Algorithm \ref{alg:blackwell-approachability based regmin} with \pomd{} with $\C$ as a decision set.
\begin{algorithm}
      \caption{Predictive Treeplex Blackwell$^+$(\ptbp)}
      \label{alg:predictive-blackwell-treeplex-two-prox-calls}
      \begin{algorithmic}[1]
     \State {\bf Input}: $\eta >0$, predictions $\bm{m}_{1},...,\bm{m}_{T}\in \R^{n+1} $
      \State {\bf Initialization}: $\hat{\bm{R}}_{1} = \bm{0} \in \R^{n+1}$
      
      \For{$t = 1, \dots, T$}
      \State  $\bm{R}_{t} \in \Pi_{\C} \left(\hat{\bm{R}}_{t} - \eta  \bm{m}_{t}\right)$
      \State $\bm{x}_{t} = \bm{R}_{t}/\langle\bm{R}_{t},\bm{a}\rangle$
      \State Observe the loss vector $\bm{\ell}_{t} \in \R^{n+1}$
      \State $\hat{\bm{R}}_{t+1} \in \Pi_{\C} \left(\hat{\bm{R}}_{t} - \eta \bm{f}(\bm{x}_{t},\bm{\ell}_{t}) \right)$
      \EndFor
\end{algorithmic}
\end{algorithm}
We start by highlighting a crucial property of \ptbp{},
{\em treeplex stepsize invariance.} The sequence of iterates $\bm{x}_{1},...,\bm{x}_{T}$ generated by Algorithm \ref{alg:predictive-blackwell-treeplex-two-prox-calls} is independent of the choice of the stepsize $\eta>0$, that only rescales the sequences $\hat{\bm{R}}_{1},...,\hat{\bm{R}}_{T}$ and $\bm{R}_{1},...,\bm{R}_{T}$, the orthogonal projection onto a cone is positively homogeneous of degree $1$: $ \Pi_{\C}(\eta \bm{z})  = \eta \Pi_{\C}(\bm{z})$ for $\eta >0 $ and $\bm{z} \in \R^{n+1},$ and the function $\bm{R} \mapsto \bm{R} / \langle \bm{R},\bm{a}\rangle$ is scale-invariant: $\frac{(\eta \bm{R})}{ \langle(\eta\bm{R}),\bm{a}\rangle}  = \frac{\bm{R}}{\langle\bm{R},\bm{a}\rangle} $ for $\eta>0$ and $\bm{R} \in \R^{n+1}$.
 We provide a rigorous statement in the following proposition and we present the proof in Appendix \ref{app:proof-prop:stepsize-independent}.
 \begin{proposition}\label{prop:stepsize-independent}
The sequence $\bm{x}_{1},...,\bm{x}_{T}$ computed by \ptbp{} after $T$ iterations is independent on the stepsize $\eta>0$.
 \end{proposition}
  Treeplex stepsize invariance is a crucial property, since in large EFGs, stepsize tuning is difficult and resource-consuming. This is the main advantage of using Blackwell approachability: running \pomd{} directly on the treeplex $\T$ does not result in a stepsize invariant algorithm, whereas \ptbp{} runs \pomd{} on $\cone(\T)$ and is stepsize invariant.
 To our knowledge, \cfrp{} and \pcfrp{} are the only other treeplex stepsize invariant algorithms for solving EFGs. In fact, they satisfy a stronger {\em infoset stepsize invariance} property: different stepsizes can be used at different infosets, and the iterates do not depend on their values.  We discuss the relation between \ptbp{} and known instantiations of Blackwell approachability over the simplex (\rmp{} and \cbap~\cite{grand2023solving}) in Appendix \ref{app:comparison rmp ptbp}. 

From Proposition \ref{prop:blackwell-approachability adversarial-regret} and the regret bounds on \pomd{} (see for instance section 3.1.1 in \citet{rakhlin2013online} or section 6 in \citet{farina2021faster}), we obtain the following proposition. 
We define $\Omega \in \R_{+}$ as $\Omega := \max_{\bm{x} \in \T} \| \bm{x}\|_{2}.$
\begin{proposition}\label{prop:ptbp-adversarial-regret}
Let $\bm{x}_{1},...,\bm{x}_{T}$ be computed by \ptbp{}. Then 
\[\max_{\hat{\bm{x}} \in \T} \sum_{t=1}^{T} \langle \bm{x}_{t}- \hat{\bm{x}},\bm{\ell}_{t}\rangle \leq \Omega \sqrt{\sum_{t=1}^{T} \| \bm{f}(\bm{x}_{t},\bm{\ell}_{t}) - \bm{m}_{t} \|_{2}^{2}}.\]
\end{proposition}
From Proposition \ref{prop:ptbp-adversarial-regret}, \ptbp{} is a regret minimizer over treeplexes, and we can combine it with the self-play framework to solve EFGs, as shown in the next corollary.
We use the notations
$d :=\max\{n,m\}+1, \hat{\Omega}  := \max \{ \max\{\| \bm{x}\|_{2},\| \bm{y}\|_{2} \}\; | \; \bm{x} \in \cX,\bm{y}\in \cY\}, \\  \| \bm{M}\|_{2} := \sup \{ \frac{\| \bm{Mv}\|_{2}}{\| \bm{v}\|_{2}} \; | \; \bm{v} \neq \bm{0}\}.$
\begin{corollary}\label{cor:PTB+ convergence NE}
    Let $(\bm{x}_{t})_{t \geq 1}$ and $(\bm{y}_{t})_{t \geq 1}$ be the sequence of strategies computed by both players employing \ptbp{} in the self-play framework, with previous losses as predictions:
    $\bm{m}^{x}_{t}  =\bm{f}(\bm{x}_{t-1},\bm{My}_{t-1}),\bm{m}^{y}_{t} = \bm{f}(\bm{y}_{t-1},-\bm{M}\tr\bm{x}_{t-1}).$
    Let $\left(\bar{\bm{x}}_{T},\bar{\bm{y}}_{T}\right) = \frac{1}{T} \sum_{t=1}^{T}\left(\bm{x}_{t},\bm{y}_{t}\right)$. Then
\[ \max_{\bm{y} \in \cY} \; \langle \bar{\bm{x}}_{T},\bm{My}\rangle - \min_{\bm{x} \in \cX} \; \langle \bm{x},\bm{M}\bar{\bm{y}}_{T}\rangle \leq \frac{\hat{\Omega}^3\sqrt{d}\sqrt{\| \bm{M}\|_{2}}}{\sqrt{T}}.\]
\end{corollary}
Finally, we can efficiently compute the orthogonal projection onto $\C$ at every iteration to implement \ptbp, since $\C$ admits the following simple formulation of as a polytope:
$\C = \{ \bm{x} \in \R^{n+1}_{+} \; | \; \sum_{a \in \A_{j}} x_{ja} = x_{p_{j}}, \forall \; j \in \J\}.$
\begin{proposition}\label{prop:practical implementation ptbp}
Let $\T$ be a treeplex with depth $d$, number of sequences $n$, number of leaf sequences $l$, and number of infosets $m$. The orthogonal projection $\Pi_{\C}(\bm{y})$ of a point $\bm{y} \in \R^{n+1}$ onto $\C=\cone(\T)$ can be computed in $O(d n \log (l+m))$ arithmetic operations.
\end{proposition}
\paragraph{A stable algorithm: \smoothptbp{}.}
We now modify \ptbp{} to obtain faster convergence rates for solving EFGs. The $O(1/\sqrt{T})$ average convergence rate of \ptbp{} may seem surprising since in the {\em matrix game} setting, \pomd{} over the simplexes obtains a $O(1/T)$ average convergence~\citep{syrgkanis2015fast}. This discrepancy comes from \ptbp{} running \pomd{} {\em on the set $\C = \cone(\T)$} instead of the original decision set $\T$, so that the Lipschtiz continuity of the loss function and the classical {\em RVU bounds} (Regret Bounded by Variation in Utilities, see Equation (1) in \citet{syrgkanis2015fast}), central to proving the fast convergence of predictive algorithms, may not hold.  
For \ptbp{}, the Lipschitz continuity of the loss $\bm{R} \mapsto \bm{f}(\bm{x},\bm{\ell})$ with $\bm{x}=\bm{R}/\langle \bm{R},\bm{a}\rangle$ depends on the Lipschitz continuity of the decision function $\bm{R} \mapsto \bm{R}/\langle\bm{R},\bm{a}\rangle$ over $\C$, which we analyze next.
\begin{proposition}\label{prop:lipschitness-chd}
Let $\bm{R}_{1},\bm{R}_{2} \in \cone(\T)$. Then
\[ \left\| \frac{\bm{R}_{1}}{\langle\bm{R}_{1},\bm{a}\rangle} - \frac{\bm{R}_{2}}{\langle\bm{R}_{2},\bm{a}\rangle} \right\|_{2} \leq \frac{ \Omega \cdot \| \bm{R}_{1} - \bm{R}_{2}\|_{2}}{\max \{ \langle\bm{R}_{1},\bm{a}\rangle,\langle\bm{R}_{2},\bm{a}\rangle\}}.\]
\end{proposition}
We present the proof of Proposition \ref{prop:lipschitness-chd} in Appendix \ref{app:proof-prop-lip-chd}.
Proposition \ref{prop:lipschitness-chd} shows that when the vector $\bm{R}$ is such that $\langle\bm{R},\bm{a}\rangle$ is small, the decision function $\bm{R} \mapsto \bm{R}/\langle\bm{R},\bm{a}\rangle$ may vary rapidly, an issue known as {\em instability} and also observed for a predictive variant of \rmp~\citep{farina2023regret}. To ensure the Lipschitzness of the decision function, we can ensure that $\bm{R}_{t}$ and $\hat{\bm{R}}_{t}$ always belong to the {\em stable region} $\C_{\geq}$:
\[\C_{\geq} := \cone(\T) \cap \{ \bm{R} \in \R^{n+1} \; | \; \langle\bm{R},\bm{a}\rangle \geq R_{0}\}\] for $R_{0}>0$, and we recover Lipschitz continuity over $\Cgeq$:
\[ \left\| \frac{\bm{R}_{1}}{\langle\bm{R}_{1},\bm{a}\rangle} - \frac{\bm{R}_{2}}{\langle\bm{R}_{2},\bm{a}\rangle} \right\|_{2} \leq \frac{ \Omega}{R_{0}} \| \bm{R}_{1} - \bm{R}_{2}\|_{2}, \forall \; \bm{R}_{1},\bm{R}_{2} \in \Cgeq.\] 
This leads us to introduce {\em Smooth Predictive Treeplex Blackwell$^+$} (\smoothptbp, Algorithm \ref{alg:pred-stable-blackwell-treeplex}), a variant of \ptbp,where $\bm{R}_{t}$ and $\hat{\bm{R}}_{t}$ always belong to $\C_{\geq}$.
\begin{algorithm}
      \caption{\smoothptbp}
      \label{alg:pred-stable-blackwell-treeplex}
      \begin{algorithmic}[1]
     \State {\bf Input}: $\eta >0$, predictions $\bm{m}_{1},...,\bm{m}_{T}\in \R^{n+1} $
      \State {\bf Initialization}: $\hat{\bm{R}}_{1} = \bm{0} \in \R^{n+1}$
      
      \For{$t = 1, \dots, T$}
      \State  $\bm{R}_{t} \in \Pi_{\Cgeq} \left(\hat{\bm{R}}_{t} - \eta  \bm{m}_{t}\right)$
      \State $\bm{x}_{t} = \bm{R}_{t}/\langle\bm{R}_{t},\bm{a}\rangle$
      \State Observe the loss vector $\bm{\ell}_{t} \in \R^{n+1}$
      \State $\hat{\bm{R}}_{t+1} \in \Pi_{\Cgeq} \left(\hat{\bm{R}}_{t} - \eta \bm{f}(\bm{x}_{t},\bm{\ell}_{t}) \right)$
      \EndFor
\end{algorithmic}
\end{algorithm}

For \smoothptbp{}, we first note that $\bm{x}_{t} \in \T$ since $\bm{R}_{t} \in \C_{\geq} \subset \cone(\T)$. We also have the hyperplane forcing property \eqref{eq:halfspace-forcing}, which only depends on $\bm{x}_{t} = \bm{R}_{t}/\langle\bm{R}_{t},\bm{a}\rangle$. However, \smoothptbp{} is not treeplex stepsize invariant, because the orthogonal projections are onto $\Cgeq$, which is not a cone.
Note that $\Cgeq$ admits a simple polytope formulation:
\[\Cgeq = \{ \bm{x} \in \R^{n+1}_{+} \; | \xnot \geq R_{0}, \sum_{a \in \A_{j}} x_{ja} = x_{p_{j}}, \forall \; j \in \J\}\]
so the complexity of computing the orthogonal projection onto $\Cgeq$ is the same as computing the orthogonal projection onto $\C$. We provide a proof in Appendix \ref{app: practical implementation smooth ptbp}.

We now show that \smoothptbp{} is a regret minimizer. Indeed, the proof of Proposition \ref{prop:blackwell-approachability adversarial-regret} can be adapted to relate the regret in $\bm{x}_{1},...,\bm{x}_{T}$ in $\T$ to regret in $\bm{R}_{1},...,\bm{R}_{T}$ in $\Cgeq$.
\begin{proposition}\label{prop:smooth-ptbp-adversarial-regret}
Let $\bm{x}_{1},...,\bm{x}_{T}$ be computed by \smoothptbp{}. For $\eta = \sqrt{2 \Omega}/\sqrt{\sum_{t=1}^{T} \| f(\bm{x}_{t},\bm{\ell}_{t}) - \bm{m}_{t} \|_{2}^{2}}$, we have
\[\max_{\hat{\bm{x}} \in \T} \sum_{t=1}^{T} \langle \bm{x}_{t}- \hat{\bm{x}},\bm{\ell}_{t}\rangle \leq \Omega \sqrt{\sum_{t=1}^{T} \| f(\bm{x}_{t},\bm{\ell}_{t}) - \bm{m}_{t} \|_{2}^{2}}.\]
\end{proposition}
Because in \smoothptbp{} $\bm{R}_{t}$ and $\hat{\bm{R}}_{t}$ always belong to $\Cgeq$, we are able to recover a RVU bound for \smoothptbp{} and show fast convergence rates.
We now define $\| \bm{M}\|$ as the maximum $\ell_{2}$-norm of any column and any row of $\bm{M}$.
\begin{theorem}\label{th:regret-guarantees-smooth-sptbp-2playerEFG}
    Let $(\bm{x}_{t})_{t \geq 1}$ and $(\bm{y}_{t})_{t \geq 1}$ be the sequence of strategies computed by both players employing \smoothptbp{} in the self-play framework, with previous losses as predictions:
    $\bm{m}^{x}_{t}  =\bm{f}(\bm{x}_{t-1},\bm{My}_{t-1}),\bm{m}^{y}_{t} = \bm{f}(\bm{y}_{t-1},-\bm{M}\tr\bm{x}_{t-1}).$
 Let $\eta = \frac{R_{0}}{\sqrt{8d\hat{\Omega}^3}\| \bm{M}\| }$
and  $\left(\bar{\bm{x}}_{T},\bar{\bm{y}}_{T}\right) = \frac{1}{T} \sum_{t=1}^{T}\left(\bm{x}_{t},\bm{y}_{t}\right)$. Then
\[ \max_{\bm{y} \in \cY} \; \langle \bar{\bm{x}}_{T},\bm{My}\rangle - \min_{\bm{x} \in \cX} \; \langle \bm{x},\bm{M}\bar{\bm{y}}_{T}\rangle \leq \frac{2 \hat{\Omega}^2}{\eta}\frac{1}{T}.\]
\end{theorem}
We present the proof of Theorem \ref{th:regret-guarantees-smooth-sptbp-2playerEFG} in Appendix \ref{app:proof convergence smooth pbtb NE}. 
To the best of our knowledge, \smoothptbp{} is the first algorithm based on Blackwell approachability achieving the state-of-the-art $O(1/T)$ convergence rate for solving \eqref{eq:EFG-BSSP}.
Other methods achieving this rate include Mirror Prox and Excessive Gap Technique for EFGs~\citep{kroer2018faster} and predictive OMD directly on the treeplex~\citep{farina2019optimistic}. 
We can compare the $O(1/T)$ convergence rate of \smoothptbp{} with the $O(1/\sqrt{T})$ convergence rate of Predictive \cfrp~\citep{farina2021faster}, which combines CFR with Predictive \rmp{} as a regret minimizer (see Appendix \ref{app:cfr}). 
Despite its predictive nature, Predictive \cfrp{} only achieves a $O(1/\sqrt{T})$ convergence rate because of the CFR decomposition, which enables running regret minimizers {\em independently and locally} at each infoset, and it is not clear how to combine, at the treeplex level, the regret bounds obtained at each infoset. Since \smoothptbp{} operates directly over the entire treeplex, we can combine the RVU bound for each player to obtain our state-of-the-art convergence guarantees.
\paragraph{An adaptive algorithm: \adatbp.}
For completeness, we now instantiate Algorithm \ref{alg:blackwell-approachability based regmin} with a regret minimizer that can learn potentially heterogeneous stepsizes across information sets in an adaptive fashion.
We introduce \adatbp (Algorithm \ref{alg:adagrad-stable-blackwell-treeplex}), a variant of Algorithm \ref{alg:blackwell-approachability based regmin} that learns a different stepsize for each coordinate of $\bm{R}_{t} \in \cone(\T)$ in an adaptive fashion inspired by the \adagrad{} algorithm~\cite{duchi2011adaptive}, which adapts the scale of the stepsizes for each dimension to the magnitude of the observed gradients for this dimension. Given matrix $\bm{A}$ and a vector $\bm{y} \in \R^{n+1}$, let $\diag(\bm{y})$ be the diagonal matrix with $\bm{y}$ on its diagonal and $\Pi_{\C}^{A}(\bm{y}) = \arg \min_{\bm{x} \in \C} \langle \bm{x} - \bm{y},\bm{A}(\bm{x} - \bm{y})\rangle.$
We first show that \adatbp{} is a regret minimizer.
\begin{proposition}\label{prop:adatbp-adversarial-regret}
Let $\bm{x}_{1},...,\bm{x}_{T}$ be computed by \adatbp{}. For $\eta = \frac{\max_{t \leq T} \left(\| \bm{R}_{t}\|_{2}+\Omega\right)^{2}} {\sqrt{2}}$, we have
\[\max_{\hat{\bm{x}} \in \T} \sum_{t=1}^{T} \langle \bm{x}_{t}- \hat{\bm{x}},\bm{\ell}_{t}\rangle \leq 2 \eta  \sum_{i=1}^{d} \sqrt{\sum_{t=1}^{T} \left(\bm{f}_{t}(\bm{x}_{t},\bm{\ell}_t)\right)_{i}^{2}}.\]
\end{proposition}
We omit the proof of Proposition \ref{prop:adatbp-adversarial-regret} for conciseness; it follows from the regret guarantees of \adagrad{} (Theorem 5 in \citet{duchi2011adaptive}) and Proposition \ref{prop:blackwell-approachability adversarial-regret}.  We conclude that combining \adatbp{} with the self-play framework ensures a $O(1/\sqrt{T})$ convergence rate. \adatbp{} learn different stepsizes across the treeplex, which may be useful if the losses across different dimensions differ in magnitudes. 
On the other hand, the stepsizes of \adatbp{} decrease over time, which could be too conservative.
\begin{algorithm}
      \caption{\adatbp}
      \label{alg:adagrad-stable-blackwell-treeplex}
      \begin{algorithmic}[1]
     \State {\bf Input}: $\eta,\delta >0$
      \State {\bf Initialization}: $\bm{R}_{1} = \bm{s}_{0} = \bm{g}_{0} = \bm{0} \in \R^{n+1}$
      \For{$t = 1, \dots, T$}
      \State $\bm{x}_{t} = \bm{R}_{t}/\langle\bm{R}_{t},\bm{a}\rangle$
      \State Observe the loss vector $\bm{\ell}_{t} \in \R^{n+1}$
      \State $\bm{s}_{t} = \bm{s}_{t-1} + \bm{f}(\bm{x}_{t},\bm{\ell}_{t}) \odot \bm{f}(\bm{x}_{t},\bm{\ell}_{t}) $
      \State $\bm{H}_{t} = \diag\left(\sqrt{\bm{s}_{t}}+\epsilon \bm{1}\right)$
      \State $\bm{R}_{t+1} \in \Pi_{\C}^{\bm{H}_{t}} \left(\bm{R}_{t} - \eta \bm{H}^{-1}_{t} \bm{f}(\bm{x}_{t},\bm{\ell}_{t}) \right)$
      \EndFor
\end{algorithmic}
\end{algorithm}
\section{Numerical Experiments}
We conduct two sets of numerical experiments to investigate the performance of our algorithms for solving several two-player zero-sum EFG benchmark games: Kuhn poker, Leduc poker, Liar's Dice, Goofspiel and Battleship. Additional experimental detail is given in Appendix \ref{app:simu}. 

\begin{figure*}[htb]
\centering
\includegraphics[width=\textwidth]{"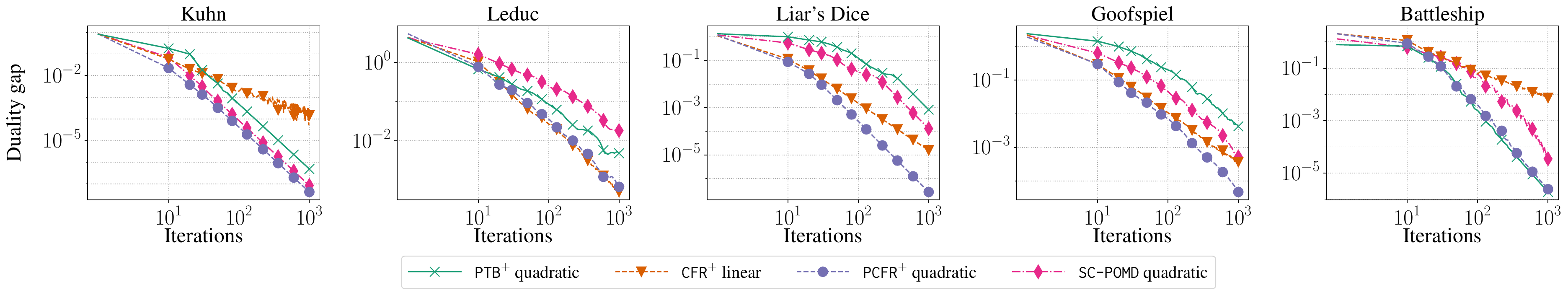"}
\caption{Convergence to Nash equilibrium as a function of number of iterations for \ptbp{} with quadratic averaging, \cfrp{} with linear averaging, \pcfrp{} with quadratic averaging, and \scpomd{} with quadratic averaging. Every algorithm is using alternation.}
\label{fig:all_algo_avg_comparison}
\end{figure*}

\begin{figure*}[htb]
\centering
\includegraphics[width=\textwidth]{"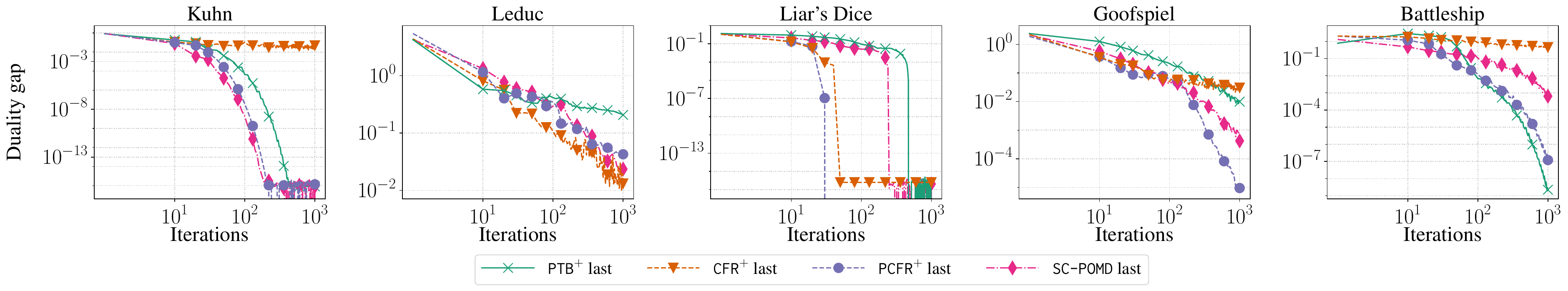"}
\caption{Convergence to Nash equilibrium for the last iterates of \ptbp{}, \cfrp{}, \pcfrp{}, and \scpomd{}. Every algorithm is using alternation.}
\label{fig:all_algo_last_comparison}
\end{figure*}

We first determine the best way to instantiate our framework. We compare \ptbp{}, \smoothptbp{} and \adatbp{} in the self-play framework with alternation (see Appendix \ref{app:self-play framework}) and uniform, linear or quadratic weights for the iterates. \ptbp{} and \smoothptbp{} use the previous losses as predictions. We also study {\em Treeplex Blackwell}$^+$ (\tbp{}), corresponding to \ptbp{} without predictions ($\bm{m}_{t}=\bm{0}$), and \adamtbp, inspired by \adam~\cite{kingma2014adam}.  For conciseness, we present our plots in Appendix \ref{fig:comparison tb algo} (Figure \ref{fig:tb_algo_comparison}) and state our conclusion here. We find that, for every game, \ptbp{} performs the best or is among the best algorithms. This underlines the advantage of {\em treeplex stepsize invariance} over algorithms that require tuning a stepsize (\smoothptbp) 
and adaptive algorithms (\adatbp), which may perform poorly due to the stepsize decreasing at a rate of $O(1/\sqrt{T})$.
\adamtbp{} does not even converge in some games. 

We then compare the best of our algorithms (\ptbp) with some of the best existing methods for solving EFGs: \cfrp~\citep{tammelin2015solving}, predictive \cfrp{} (\pcfrp, \citet{farina2021faster}, see Appendix \ref{app:cfr}), and a version of optimistic online mirror descent with a single call to the orthogonal projection at every iteration (\scpomd,~\citet{joulani2017modular}) achieving a $O(1/T)$ convergence rate; there are a variety of FOMs with a $O(1/T)$ rate, \scpomd{} was observed to perform well in \citet{chakrabarti2023block}. 
We determine the best empirical setup for each algorithm in Appendix \ref{app:individual performance}. 
In Figure \ref{fig:all_algo_avg_comparison}, we compare the performance of the (weighted) average iterates. 
We find that \pcfrp{} outperforms both \cfrp{} and the theoretically-faster \scpomd{}, as expected from past work. 
We had hoped to see at least comparable performance between \ptbp{} and \pcfrp{}, since they are both based on Blackwell-approachability regret minimizers derived from applying \pomd{} on the conic hull of their respective decision sets (simplexes at each infoset for \pcfrp, treeplexes of each player for \ptbp). However, in some games \pcfrp{} performs much better than \ptbp.
Given the similarity between \ptbp{} and \pcfrp{}, our results suggest that the use of the CFR decomposition is part of the key to the performance of \pcfrp{}. In particular, the CFR decomposition allows \pcfrp{} to have stepsize invariance {\em at an infoset level}, as opposed to stepsize invariance at the treeplex level in \ptbp{}.
Because of the structure of treeplexes, the numerical values of variables associated with infosets appearing late in the game, i.e., deeper in the treeplexes, may be much smaller than the numerical values of the variables appearing closer to the root. For this reason, allowing for different stepsizes at each infosets (like \cfrp{} and \pcfrp{} do) appears to be more efficient than using a single stepsize across all the infosets, even when the iterates do not depend on the value of this single stepsize (like in \ptbp) and when this stepsize is fine-tuned (like in \scpomd). 
Of course one could try to run \scpomd{} with different stepsizes at each infoset and attempt to tune each of these stepsizes, but this is impossible in practical instances where the number of actions is large, e.g., $4.9 \times 10^4$ actions in {\em Liar's Dice} and $5.3 \times 10^6$ actions in {\em Goofspiel}. \cfrp{} and \pcfrp{} bypass this issue with their infoset stepsize invariance, which enables both each infoset to have its own stepsize (via the CFR decomposition) {\em and} not needing to choose these stepsizes (via using \rmp{} and \prmp{} as local regret minimizers,  which are stepsize invariant).

We also investigate the performance of the {\em last iterates} in Figure \ref{fig:all_algo_last_comparison}. No algorithm appears to be the best across all game instances. \cfrp{} may not converge to a Nash equilibrium (e.g., on Kuhn), as has been observed before~\citep{lee2021last}. \pcfrp{} exhibits linear convergence in some games (Kuhn, Liar's Dice, Goofspiel) but not others (Leduc). The same is true for \ptbp. Further investigations about last-iterate convergence are left as an important open question.
\section{Conclusion}
We propose the first Blackwell approachability-based regret minimizer over the treeplex (Algorithm \ref{alg:blackwell-approachability based regmin}) and we give several instantiations of our framework with different properties, including treeplex stepsize invariance (\ptbp), adaptive stepsizes (\adatbp) and achieving state-of-the-art $O(1/T)$ convergence guarantees on EFGs with a Blackwell approachability-based algorithm for the first time (\smoothptbp). 
Since \cfrp{} and \pcfrp{} are stepsize invariant and have strong empirical performance, we were expecting \ptbp{} to have comparable performance. However, our experiments show that \ptbp{} often converges slower than \cfrp{} and \pcfrp{}, so this stepsize invariance is not the only driver behind the practical performance of \cfrp{} and \pcfrp. We view this negative result as an important contribution of our paper, since it rules out a previously plausible explanation for the practical performance of \cfrp.  Instead, we propose that one piece of the puzzle behind the \cfrp{} and \pcfrp{} performances is their {\em infoset} stepsize invariance, a consequence of combining the CFR framework with Blackwell approachability-based regret minimizers (\rmp{} and \prmp, themselves stepsize invariant over simplexes).
Future works include better understanding the last-iterate performance of algorithms based on Blackwell approachability as well as the role of alternation.


\newpage

\paragraph{Acknowledgments}
Darshan Chakrabarti was supported by the National Science Foundation Graduate Research Fellowship Program under award number DGE-2036197.
Julien Grand-Cl{\'e}ment was supported by Hi! Paris and by a grant of the French National Research Agency (ANR), “Investissements d’Avenir” (LabEx Ecodec/ANR-11-LABX-0047).
Christian Kroer was supported by the Office of Naval Research awards N00014-22-1-2530 and N00014-23-1-2374, and the National Science Foundation awards IIS-2147361 and IIS-2238960.


\bibliographystyle{plainnat} 
\bibliography{ref}


\appendix
\section{Self-Play Framework}\label{app:self-play framework}
The (vanilla) self-play framework for two-player zero-sum EFGs is presented in Algorithm \ref{alg:self-play framework}.
\begin{algorithm}
      \caption{self-play framework}
      \label{alg:self-play framework}
      \begin{algorithmic}[1]
     \State {\bf Input}: $\Regmin_{\cX}$ a regret minimizer over $\cX$, $\Regmin_{\cY}$ a regret minimizer over $\cY$      
      \For{$t = 1, \dots, T$}
      \State $\bm{x}_{t} = \Regmin_{\cX}\left(\cdot \right)$
      \State $\bm{y}_{t} = \Regmin_{\cY}\left(\cdot \right)$
      \State The first player observes the loss vector $\bm{Ay}_{t} \in \R^{n_{1}+1}$
      \State The second player observes the loss vector $-\bm{A}\tr\bm{x}_{t} \in \R^{n_{2}+1}$
      \EndFor
\end{algorithmic}
\end{algorithm}
The self-play framework can be combined with {\em alternation}, a simple variant that is known to lead to significant empirical speedups, for instance, when \cfrp{} and predictive \cfrp{} are used as regret minimizers~\citep{tammelin2015solving,farina2021faster,burch2019revisiting}. When using alternation, at iteration $t$ the second player is provided with the current strategy of the first player $\bm{x}_t$ before choosing its own strategy. We describe the self-play framework with alternation in Algorithm \ref{alg:self-play framework alternation}.
\begin{algorithm}
      \caption{self-play framework with alternation}
      \label{alg:self-play framework alternation}
      \begin{algorithmic}[1]
     \State {\bf Input}: $\Regmin_{\cX}$ a regret minimizer over $\cX$, $\Regmin_{\cY}$ a regret minimizer over $\cY$      
      \For{$t = 1, \dots, T$}
      \State $\bm{x}_{t} = \Regmin_{\cX}\left(\cdot \right)$
      \State The second player observes the loss vector $-\bm{A}\tr\bm{x}_{t} \in \R^{n_{2}+1}$
      \State $\bm{y}_{t} = \Regmin_{\cY}\left(\cdot \right)$
      \State The first player observes the loss vector $\bm{Ay}_{t} \in \R^{n_{1}+1}$
      \EndFor
\end{algorithmic}
\end{algorithm}
\section{Counterfactual Regret Minimization (CFR), \cfrp{} and Predictive \cfrp}\label{app:cfr}
Counterfactual Regret Minimization (CFR, \cite{zinkevich2007regret}) is a framework for regret minimization over the treeplex. CFR runs a regret minimizer $\Regmin_{j}$ locally at each infoset $j \in \J$ of the treeplex. Note that here $\Regmin_{j}$ is a regret minimizer over the {\em simplex} $\Delta^{n_{j}}$ with $n_{j}=|\A_{j}|$, i.e., over the set of probability distributions over $\A_{j}$, the set of actions available at infoset $j \in \J$. Let $\bm{x}^{j}_{t} \in \Delta^{n_{j}}$ be the decision chosen by $\Regmin_{j}$ at iteration $t$ in CFR and let $\bm{\ell}_{t} \in \R^{n+1}$ be the loss across the entire treeplex. The {\em local loss} $\bm{\ell}_{t}^{j} \in \R^{n_{j}}$ that CFR passes to $\Regmin_{j}$ is
\[\ell_{t,a}^{j} := \ell_{t,(j,a)} + \sum_{j' \in \C_{ja}} V_{t}^{j'}, \forall \; a \in \A_{j}, \forall \; j \in \J\]
where $V^{j}_{t}$ is the {\em value function} for infoset $j$ at iteration $t$, defined inductively:
\[ V^{j}_{t} := \sum_{a \in \A_{j}} x^{j}_{t,a}\ell_{t,(j,a)} + \sum_{j' \in \C_{ja}} V_{t}^{j'}.\]
The regret over the entire treeplex $\T$ can be related to the regrets accumulated at each infoset via the following {\em laminar regret decomposition}~\cite{farina2019online}:
\[\Reg^{T}:=\max_{\hat{\bm{x}} \in \T} \sum_{t=1}^{T} \langle \bm{x}_{t}- \hat{\bm{x}},\bm{\ell}_{t}\rangle \leq \max_{\hat{\bm{x}} \in \T} \sum_{j \in \J} \hat{x}_{p_{j}} \Reg^{T}_{j}\left(\hat{\bm{x}}^{j}\right)\]
with $\Reg^{T}_{j}\left(\hat{\bm{x}}^{j}\right):= \sum_{t=1}^{T} \langle \bm{x}_{t}^{j}- \hat{\bm{x}}^{j},\bm{\ell}_{t}^{j}\rangle$ the regret incured by $\Regmin^{j}$ for the sequence of losses $\bm{\ell}_{1}^{j},...,\bm{\ell}_{T}^{j}$ against the comparator $\hat{\bm{x}}^{j} \in \Delta^{n_{j}}$. 
Combining CFR with regret minimizers at each information set ensures $\Reg^{T}=O\left(\sqrt{T}\right)$.

\cfrp~\citep{tammelin2015solving} corresponds to instantiating the self-play framework with alternation (Algorithm \ref{alg:self-play framework alternation}) and Regret Matching$^+$ (\rmp{} as presented in \eqref{eq:rmp}) as a regret minimizer at each information set. Additionally, \cfrp{} uses {\em linear} averaging, i.e., it returns $\bar{\bm{x}}_T$ such that $\bar{\bm{x}}_T = \frac{1}{\sum_{t=1}^{T}\omega_t} \sum_{t=1}^{T} \omega_t \bm{x}_t$ with $\omega_t=t$. We also consider uniform weights ($\omega_t=1$) and quadratic weights ($\omega=t^2$) in our simulations (Figure \ref{fig:cfrp}). \cfrp{} guarantees a $O(1/\sqrt{T})$ convergence rate to a Nash equilibrium.

Predictive \cfrp (\pcfrp, \cite{farina2021faster}) corresponds to instantiating the self-play framework with alternation (Algorithm \ref{alg:self-play framework alternation}) and Predictive Regret Matching$^+$ (\prmp) as a regret minimizer at each information set. Given a simplex $\Delta^d$, \prmp{} is a regret minimizer that returns a sequence of decisions $\bm{z}_{1},...,\bm{z}_{T} \in \Delta^d$ as follows:
\begin{equation}\label{eq:prmp}\tag{\prmp}
    \begin{aligned}
    \hat{\bm{R}}_{t} & = \Pi_{\R_{+}^{d}}\left(\bm{R}_{t} - \eta \bm{g}(\bm{z}_{t-1},\bm{\ell}_{t-1})\right) \\
    \bm{z}_{t} & = \hat{\bm{R}}_{t} / \| \hat{\bm{R}}_{t}\|_{1},\\ \bm{R}_{t+1} & = \Pi_{\R_{+}^{d}}\left(\bm{R}_{t} - \eta \bm{g}(\bm{z}_{t},\bm{\ell}_{t})\right)
\end{aligned}
\end{equation}

where the function $g$ is defined in \eqref{eq:definition g}.
Similar to \cfrp, for \pcfrp{} we investigate different weighting schemes in our numerical experiments (Figure \ref{fig:pcfrp}). It is not known if the self-play framework with alternation, combined with \pcfrp, has convergence guarantees, but \pcfrp{} has been observed to achieve state-of-the-art practical performance in many EFG instances~\cite{farina2021faster}.

\section{Proof of Proposition \ref{prop:stepsize-independent}}\label{app:proof-prop:stepsize-independent}
\begin{proof}

The proof of Proposition \ref{prop:stepsize-independent} is based on the following lemma.
\begin{lemma}\label{lem:homogeneous-conic-projection}
Let $\C \subset \R^{n}$ be a convex cone and let $\bm{u} \in \R^{n},\eta>0$. Then
\[ \Pi_{\C}(\eta \bm{u}) = \eta \Pi_{\C}(\bm{u}).\]
\end{lemma}
\begin{proof}[Proof of Lemma \ref{lem:homogeneous-conic-projection}]
We have, by definition,
\[ \Pi_{\C}(\eta \bm{u})  = \arg \min_{\bm{R} \in \C} \| \bm{R} - \eta \bm{u} \|_{2}.\] 
Now we also have
\[ \min_{\bm{R} \in \C} \| \bm{R} - \eta \bm{u} \|_{2}  =\eta \cdot \min_{\bm{R} \in \C}  \| \frac{1}{\eta}\bm{R} - \bm{u} \|_{2} = \eta \cdot \min_{\bm{R} \in \C}  \| \bm{R} - \bm{u} \|_{2} \]
where the last equality follows from $\C$ being a cone. This shows that $\arg \min_{\bm{R} \in \C} \| \bm{R} - \eta \bm{u} \|_{2}$ is attained at $\eta \Pi_{\C}(\bm{u})$, i.e., that $\Pi_{\C}(\eta \bm{u}) = \eta \Pi_{\C}(\bm{u})$.
\end{proof}
We are now ready to prove Proposition \ref{prop:stepsize-independent}. For the sake of conciseness we prove this with $\bm{m}_{1} = ... = \bm{m}_{T}= \bm{0}$; the proof for \ptbp{} with predictions is identical. In this case, $\bm{R}_{t}=\hat{\bm{R}}_{t}, \forall \; t \geq 1$.
Consider the sequence of strategies $\tilde{\bm{x}}_{1},...,\tilde{\bm{x}}_{T}$ and $\tilde{\bm{R}}_{1},...,\tilde{\bm{R}}_{T}$ generated by \ptbp{} with a step size of $1$. We also consider the sequence of strategies $\bm{x}_{1},...,\bm{x}_{T}$ and $\bm{R}_{1},...,\bm{R}_{T}$ generated with a step size $\eta >0$. We claim that \[\tilde{\bm{x}}_{t} = \bm{x}_{t},\bm{R}_{t} = \eta \tilde{\bm{R}}_{t}, \; \forall \; t \in \{1,...,T\}.\]
We prove this by induction. Both sequences of iterates are initialized with $\bm{R}_{1} = \tilde{\bm{R}}_{1} = \bm{0}$ so that $\tilde{\bm{x}}_{1} = \bm{x}_{1}$. Therefore, both sequences face the same loss $\bm{\ell}_{1}$ at $t=1$, and we have
\[ \bm{R}_{2} = \Pi_{\C}(-\eta\bm{f}(\bm{x}_{1},\bm{\ell}_{1})) = \eta \pi_{\C}(-\bm{f}(\bm{x}_{1},\bm{\ell}_{1}))) = \eta \tilde{\bm{R}}_{2}.\]
Let us now consider an iteration $t \geq 1$ and suppose that $\tilde{\bm{x}}_{t} = \bm{x}_{t}, \bm{R}_{t} = \eta \tilde{\bm{R}}_{t}$. Since $\tilde{\bm{x}}_{t} = \bm{x}_{t}$ then both algorithms will face the next loss vector $\bm{\ell}_{t}$. Then
\begin{align*}
\bm{R}_{t+1} & = \pi_{\C}(\bm{R}_{t} - \eta \bm{f}(\bm{x}_{t},\bm{\ell}_{t})) \\
& = \pi_{\C}(\eta \tilde{\bm{R}}_{t} -\eta \bm{f}(\bm{x}_{t},\bm{\ell}_{t}))  \\
& = \eta \pi_{\C}(\tilde{\bm{R}}_{t} - \bm{f}(\bm{x}_{t},\bm{\ell}_{t}))  \\
& = \eta \tilde{\bm{R}}_{t+1}
\end{align*}
which in turns implies that $\bm{x}_{t+1} = \tilde{\bm{x}}_{t+1}$. We conclude that $\bm{x}_{t} = \tilde{\bm{x}}_{t}, \forall \; t = 1,...,T$. 
\end{proof}

\section{Comparison Between \rmp{} and \ptbp}\label{app:comparison rmp ptbp}
Assume that the original decision set of each player is a simplex $\Delta^{d}$ and that there are no predictions: $\bm{m}_{t} = \bm{0},\forall \; t \geq 1$.

\paragraph{\ptbp{} over the simplex.} For \ptbp, the empty sequence variable $\xnot$ is introduced and appended to the decision $\Delta^d$. The resulting treepplex can be written $\T = \{1\} \times \Delta^{d}$, the set $\C$ becomes $\C:= \cone(\T) = \cone(\{1\} \times \Delta^{d})$ and $\bm{a} = (1,\bm{0}) \in \R^{d+1}_{+}$ with $1$ on the first component related to $\xnot$ and $0$ everywhere else. 
In this case, \ptbp{} without prediction is exactly the {\em Conic Blackwell Algorithm}$^+$  (\cbap, \citet{grand2023solving}). Crucially, to run \ptbp{} we need to compute the orthogonal projection onto $\cone(\T)=\cone(\{1\} \times \Delta)$ at every iteration, which can not be computed in closed-form, but it can be computed in $O(n\log(n))$ arithmetic operations (see Appendix G.1 in \cite{grand2023solving}).

\paragraph{Regret Matching$^+$.} \rmp{} operates directly over the simplex $\Delta^d$ without the introduction of the empty sequence $\xnot$, in contrast to \ptbp{} which operates over $\{1\}\times \Delta^d$. Importantly, in \rmp{}, at every iteration the orthogonal projection onto $\R^{d}_{+}$ can be computed in closed form by simply thresholding to zero the negative components (and leaving unchanged the positive components): $\Pi_{\R_{+}^{d}}\left(\bm{z}\right) = \left(\max\{z_{i},0\}\right)_{i \in [d]}$ for any $\bm{z} \in \R^{d}$.

\paragraph{Empirical comparisons.}
The numerical experiments in \cite{grand2023solving} show that \cbap{} may be slightly faster than \rmp{} for some matrix games in terms of speed of convergence as a function of the number of iterations, but it can be slower in running times because of the orthogonal projections onto $\cone(\{1\} \times \Delta)$ at each iteration (Figures 2,3,4 in \cite{grand2023solving}). When $\T$ is a treeplex that is not the simplex, introducing $\xnot$ also changes the resulting algorithm but not the complexity of the orthogonal projection onto $\cone(\T)$, since there is no closed-form anymore, even without $\xnot$. As a convention, in this paper, we will always use $\xnot$ in our description of treeplexes and of our algorithms since it is convenient from a writing and implementation standpoint. 

Overall, we notice that in the case of the simplex introducing the empty sequence variable $\xnot$ radically alters the complexity per iterations and the resulting algorithm, a fact that has not been noticed in previous work.

\section{Proof of Proposition \ref{prop:lipschitness-chd}}\label{app:proof-prop-lip-chd}
\begin{proof}[Proof of Proposition \ref{prop:lipschitness-chd}]
\begin{enumerate}
\item 
Let $ \hat{\bm{R}}_{2} = \bm{R}_{2}/\| \bm{R}_{2}\|_{2}$ be the unit vector pointing in the same direction as $\bm{R}_{2}$ and let $\bm{h} := \left(\langle\bm{R}_{1},\hat{\bm{R}_{2}}\rangle\right)\hat{\bm{R}_{2}}$ the orthogonal projection of $\bm{R}_{1}$ onto $\{ \alpha \hat{\bm{R}}_{2} \; | \; \alpha \in \R\}$. We thus have $\| \bm{R}_{1} - \bm{R}_{2} \|_{2} \geq \| \bm{R}_{1} - \bm{h}\|_{2}$.
\item Let $\bm{p} = \frac{\langle \bm{R}_{1},\bm{a}\rangle}{\langle \bm{R}_{2},\bm{a}\rangle}\hat{\bm{R}}_{2}$. Since $\bm{p}$ and $\bm{R}_{2}$ are colinear, we have
\[ \left\| \frac{\bm{R}_{1}}{\langle \bm{R}_{1},\bm{a}\rangle} - \frac{\bm{R}_{2}}{\langle \bm{R}_{2},\bm{a}\rangle} \right\|_{2} = \left\| \frac{\bm{R}_{1}}{\langle \bm{R}_{1},\bm{a}\rangle} - \frac{\bm{p}}{\langle\bm{p},\bm{a}\rangle} \right\|_{2}.\]
Additionally, by construction, $\langle\bm{p},\bm{a}\rangle = \langle \bm{R}_{1},\bm{a}\rangle$, so that we obtain
\[ \left\| \frac{\bm{R}_{1}}{\langle \bm{R}_{1},\bm{a}\rangle} - \frac{\bm{R}_{2}}{\langle \bm{R}_{2},\bm{a}\rangle} \right\|_{2} = \left\| \frac{\bm{R}_{1}}{\langle \bm{R}_{1},\bm{a}\rangle} - \frac{\bm{p}}{\langle \bm{R}_{1},\bm{a}\rangle} \right\|_{2} = \frac{1}{\langle \bm{R}_{1},\bm{a}\rangle} \| \bm{R}_{1}-\bm{p}\|_{2}.\]
Note that $\langle \bm{R}_{1},\bm{a}\rangle \geq 0$ since $\bm{R}_{1} \in \cone(\T)$ and $\T \subset \{ \bm{x} \in \R^{n+1} \; | \; \langle\bm{x},\bm{a}\rangle=1\}$.
Assume that we can compute $D >0$ such that $\frac{\| \bm{R}_{1} - \bm{p} \|_{2}}{\| \bm{R}_{1} - \bm{h}\|_{2}} \leq D$. Then we have
\[\left\| \frac{\bm{R}_{1}}{\langle \bm{R}_{1},\bm{a}\rangle} - \frac{\bm{R}_{2}}{\langle \bm{R}_{2},\bm{a}\rangle} \right\|_{2} \leq \frac{D}{\langle \bm{R}_{1},\bm{a}\rangle} \| \bm{R}_{1} - \bm{h} \|_{2} \leq \frac{D}{\langle \bm{R}_{1},\bm{a}\rangle} \| \bm{R}_{1} - \bm{R_{2}} \|_{2}.\]
\item The rest of this proof focuses on showing that $\frac{\| \bm{R}_{1} - \bm{p} \|_{2}}{\| \bm{R}_{1} - \bm{h}\|_{2}} \leq \Omega$ with $\Omega = \max \{ \| \bm{x} \|_{2} | \; \bm{x} \in \T\}$.
Note that $\langle \bm{R}_{1} - \bm{p},\bm{a} \rangle = 0$. Therefore, $\frac{1}{\|  \bm{R}_{1} - \bm{p} \|_{2}} \left(\bm{R}_{1} - \bm{p} \right)$ and $\frac{1}{\|\bm{a}\|_{2}}\bm{a}$ can be completed to form an orthonormal basis of $\R^{n}$. In this basis, we have
\[ \| \hat{\bm{R}_{2}} \|_{2}^{2} \geq \frac{\left( \langle \bm{R}_{1}-\bm{p},\hat{\bm{R}_{2}} \rangle \right)^2}{\| \bm{R}_{1}-\bm{p}\|_{2}^{2}} + \frac{\left(\langle \bm{a},\hat{\bm{R}_{2}} \rangle \right)^2}{\|\bm{a}\|_{2}^{2}}.\]
Note that by construction we have $\| \hat{\bm{R}_{2}} \|_{2}^{2}=1$. Additionally, $\bm{R}_{2} \in \cone(\T)$ so that there exists $\alpha>0$ and $\bm{y} \in \T$ such that $\bm{R}_{2} = \alpha \bm{x}$. By construction of $\hat{\bm{R}}_{2}$, we have
$\hat{\bm{R}}_{2} = \frac{\alpha \bm{x} }{\| \alpha \bm{x} \|_{2}} = \frac{\bm{x}}{\| \bm{x}\|_{2}}$ and $\langle\bm{x},\bm{a}\rangle=1$. This shows that
\[ \frac{\left(\langle \bm{a},\hat{\bm{R}_{2}} \rangle \right)^2}{\|\bm{a}\|_{2}^{2}} = \frac{\left(\langle\bm{a},\bm{x}\rangle\right)^2}{\| \bm{a}\|_{2}^{2}\| \bm{x}\|_{2}^{2}} = \frac{1}{\| \bm{a}\|_{2}^{2}\| \bm{x}\|_{2}^{2}} \geq \frac{1}{ \Omega\| \bm{a}\|_{2}^{2}} \]
with $\Omega = \max \{ \| \bm{x} \|_{2} | \; \bm{x} \in \T\}$. Recall that we have chosen $\bm{a} = (1,\bm{0})$ so that $\| \bm{a} \|_{2}=1$. Overall, we have obtained
\[ 1-\frac{1}{\Omega^2} \geq \frac{\left( \langle \bm{R}_{1}-\bm{p},\hat{\bm{R}_{2}} \rangle \right)^2}{\| \bm{R}_{1}-\bm{p}\|_{2}^{2}}.\]
From the definition of the vectors $\bm{p},\bm{h}$ and $\hat{\bm{R}_{2}}$, we have 
\[\frac{\left( \langle \bm{R}_{1}-\bm{p},\hat{\bm{R}_{2}} \rangle \right)^2}{\| \bm{R}_{1}-\bm{p}\|_{2}^{2}} = \frac{\| \bm{p}-\bm{h}\|_{2}^{2}}{\| \bm{R_{1}}-\bm{p}\|_{2}^{2}}.\]
Hence, we have
\[ \| \bm{p}-\bm{h}\|_{2}^{2} \leq \left(1-\frac{1}{\Omega^2} \right) \| \bm{R}_{1} - \bm{p} \|_{2}^{2}.\]
This shows that $\| \bm{R}_{1} - \bm{h} \|_{2}^{2} \geq \frac{1}{\Omega^2}\| \bm{x}-\bm{p}\|_{2}^{2}$.
\item We conclude that 
\[ \left\| \frac{\bm{R}_{1}}{\langle \bm{R}_{1},\bm{a}\rangle} - \frac{\bm{R}_{2}}{\langle \bm{R}_{2},\bm{a}\rangle} \right\|_{2} \leq \frac{ \Omega}{\max \{ \langle \bm{R}_{1},\bm{a}\rangle,\langle \bm{R}_{2},\bm{a}\rangle\}} \| \bm{R}_{1} - \bm{R}_{2}\|_{2}.\]
\end{enumerate}
\end{proof}
\section{Proof of \cref{prop:practical implementation ptbp}}\label{app:practical implementation smooth ptbp}
In this section we show how to efficiently compute the orthogonal projection onto the cone $\C:= \cone(\T)$.
We start by reviewing the existing methods for computing the orthogonal projection onto the treeplex $\T$. This is an important cornerstone of our analysis, since the treeplex $\T$ and the cone $\C$ share an analogous structure:
\begin{align*}
    \T & = \{ \bm{x} \in \R^{n+1}_{+} \; | \; \xnot = 1, \sum_{a \in \A_{j}} x_{ja} = x_{p_{j}}, \forall \; j \in \J\} \\
     \C & = \{ \bm{x} \in \R^{n+1}_{+} \; | \; \sum_{a \in \A_{j}} x_{ja} = x_{p_{j}}, \forall \; j \in \J\}.
\end{align*}
\citet{gilpin2012first} were the first to show an algorithm for computing Euclidean projection onto the treeplex. They do this by defining a value function for the projection of a given point $\bm{y}$ onto the closed and convex \emph{scaled} set $t \Z$, letting it be half the squared distance between $\bm{y}$ and $t\Z$, for $t \in \R_{>0}$: 
\[v_{\Z}(t, \bm{y}) := \frac{1}{2} \min_{\bm{z} \in t\Z} \| \bm{z} - \bm{y}\|^2_{2}. \]
\citet{gilpin2012first} show how to recursively compute $\lambda_{\Z}(t, \bm{y})$, the derivative of this function with respect to $t$,  for a given treeplex, since treeplexes can be constructed recursively using two operations: branching and Cartesian product. In the first case, given $k$ treeplexes $\Z_1, \dots, \Z_k$, then $\Z = \{\bm{x}, \bm{x}[1]\bm{z}_1, \dots, \bm{x}[k]\bm{z}_k : x \in \Delta_k, \bm{z}_i \in \Z_i \forall i \in [k]\}$ is also a treeplex. In the second case, given $k$ treeplexes $\Z_1, \dots, \Z_k$, then $\Z = \Z_1 \times \cdots \times \Z_k$ is also a treeplex. In fact, letting the empty set be a treeplex as a base case, all treeplexes can be constructed in this way. 

However, \citet{gilpin2012first} did not state the total complexity of computing the projection, instead only stating the complexity of computing $\lambda_{\Z}(t, \bm{y})$ given the corresponding $\lambda_{\Z_i}(t, \bm{y}_i)$ functions for the treeplexes $\Z_i$ that are used to construct $\Z$ using $i \in [k]$. They state that this complexity is $O(n \log n)$, where $n$ is the number of sequences in $\Z$. Their analysis involves showing that the function $t \mapsto \lambda_{\Z}(t, \bm{y})$ is piecewise linear.

\citet{farina2022near} also consider this problem, generalizing the problem to weighted projection on the scaled treeplex, by adding an additional positive parameter $\bm{w} \in \R^n_{>0}$:
\[v_{\Z}(t, \bm{y}, \bm{w}) := \frac{1}{2} \min_{\bm{z} \in t\Z} \sum_{i = 1}^{n} \left(\bm{z}[i] - \frac{\bm{y}[i]}{\bm{w}[i]}\right)^2. \]
They do a similar analysis to \citet{gilpin2012first}, by showing how to compute the derivative $\lambda_{\Z}(t, \bm{y}, \bm{w})$ of $v_{\Z}(t, \bm{y}, \bm{w})$ with respect to $t$  recursively. They show that $t \mapsto \lambda_{\Z}(t, \bm{y}, \bm{w})$ are strictly-monotonically-increasing piecewise-linear (SMPL) functions. We will follow the analysis in \citet{farina2022near}, letting $\bm{w} = \mathbf{1}$.  

We first define a standard representation of a SMPL function.
\begin{definition}[\citep{farina2022near}]
Given a SMPL function $f$, a standard representation is an expression of the form 
\[ f(x) = \zeta + \alpha_0 x + \sum_{s=1}^{S} \alpha_s \max \{0, x - \beta_s\}\] valid for all $x \in \operatorname{dom}(f)$, $S \in \N \cup \{0\}$, and $\beta_1 < \cdots < \beta_S$. The size of the standard representation is defined to be $S$.
\end{definition}

Next, we prove the following lemma, showing the computational complexity of computing the derivative of the value function for a given treeplex.

\begin{lemma}
For a given treeplex $\Z$ with depth $d$, $n$ sequences, $l$ leaf sequences, and $m$ infosets, and $\bm{y} \in \R^n, \bm{w} \in \R^n_{>0}$, a standard representation of $\lambda_{\Z}(t, \bm{y}, \bm{w})$ can be computed in $O\big(d n \log (l+m) \big)$ time.
\label{lemma:lambda_fn_lemma}
\end{lemma}

\begin{proof}
We will proceed by structural induction over treeplexes, following the analysis done by \citet{farina2022near}. The base case is trivially true, because the empty set has no sequences or depth. 

For the inductive case, we will assume that it requires $O\big((d-1) n \log (l+m)\big)$ time to compute the respective Euclidean projections onto the subtreeplexes that we use to inductively construct our current treeplex, where $d-1$ is the depth of a given subtreeplex, $n$ is the number of sequences in the subtreeplex, and $m$ is the total number of sequences among both players and chance corresponding to the game from which the treeplex originates.

We will use two results shown in Lemma 14 of \citet{farina2022near}:
\begin{lemma}[Recursive complexity of Euclidean projection for branching operation \citep{farina2022near}]
Consider a treeplex $\Z$ that can be written as the result of a branching operation on $k$ treeplexes $\Z_1, \dots, \Z_k$: 
\[\Z = \{\bm{x}, \bm{x}[1]\bm{z}_1, \dots, \bm{x}[k]\bm{z}_k : x \in \Delta_k, \bm{z}_i \in \Z_i \forall i \in [k]\}.\] Let $\Z$ have $n$ sequences and let $\bm{y}, \bm{w} \in \R^n$, and let $\bm{y}[i]$ and $\bm{w}[i]$ denote the corresponding respective components of $\bm{y}$ and $\bm{w}$ for the treeplex $\Z_i$.

Then, given standard representations of $\lambda_{\Z_i}(t, \bm{y}_i, \bm{w}_i)$ of size $n_i$ for all $i \in [k]$, where $n_i$ is the number of sequences that $\Z_i$ has, a standard representation of  $\lambda_{\Z}(t, \bm{y}, \bm{w})$ of size $n$ can be computed in $O(n \log k)$ time.

Furthermore, given a value of $t$, the argument $\bm{x}$ which leads to the realization of the optimal value of the value function, can be computed in time $O(n)$.
\label{lemma:projection_branching}
\end{lemma}

\begin{lemma}[Recursive complexity of Euclidean projection for Cartesian product \citep{farina2022near}]
Consider a treeplex $\Z$ that can be written as a Cartesian product of  $k$ treeplexes $\Z_1, \dots, \Z_k$: 
\[\Z = \Z_1 \times \cdots \times \Z_k.\] Let $\Z$ have $n$ sequences and let $\bm{y}, \bm{w} \in \R^n$, and let $\bm{y}[i]$ and $\bm{w}[i]$ denote the corresponding respective components of $\bm{y}$ and $\bm{w}$ for the treeplex $\Z_i$.

Then, given standard representations of $\lambda_{\Z_i}(t, \bm{y}_i, \bm{w}_i)$ of size $n_i$ for all $i \in [k]$, where $n_i$ is the number of sequences that $\Z_i$ has, a standard representation of  $\lambda_{\Z}(t, \bm{y}, \bm{w})$ of size $n$ can be computed in $O(n \log k)$ time.
\label{lemma:projection_cartesian}
\end{lemma}

First, we consider the case that the last operation used to construct our treeplex was the branching operation. Let the root of of the treeplex be called $j$. Define $\Z_i$ as the treeplex that is underneath action $a_i \in \A_j$.
Let $n_i$ denote the number of sequences in $\Z_i$, $m_i$ denote the number of infosets in $\Z_i$, $l_i$ denote the number of leaf sequences in $\Z_i$, and $d-1$ be the maximum depth of any of these subtreeplexes.

Given a standard representation of $\lambda_{\Z_i}(t, \bm{y}_i, \bm{w}_i)$ of size $n_i$ for all $i \in [|\A_j|]$, by \cref{lemma:projection_branching}, it takes $O(n \log |\A_j|)$ time to compute a standard representation of $\lambda_{\Z}(t, \bm{y}, \bm{w})$ of size $n$. By induction, it takes $O\big((d-1) n_i \log m_i\big)$ to compute $\lambda_{\Z_i}(t, \bm{y}_i, \bm{w}_i)$ for treeplex $\Z_i$. Thus the total computation required to compute $\lambda_{\Z}(t, \bm{y}, \bm{w})$ is 

\begin{align*}
O(n \log |\A_j|) + \sum_{i \in [|\A_j|]} O \big((d-1) n_i \log (l_i + m_i) \big) 
&= O(n \log |\A_j|) + \sum_{i \in [|\A_j|]} O \big((d-1) n_i \log (l + m) \big) \\
&= O(n \log |\A_j|) +  O\big((d-1) \sum_{i \in [|\A_j|]} n_i \log (l + m) \big) \\
&= O(n \log |\A_j|) + O\big((d-1) n \log (l + m) \big) \\
&= O\big(n \log (l+m)\big) + O\big((d-1) n \log (l + m) \big)\\
&= O\big(d n \log (l+m)\big)
\end{align*}
since we have necessarily that $l_i \leq l$ and $m_i \leq m$ for all $i \in [|\A_j|]$, $ \sum_{i \in [|\A_j|]} n_i \leq n$, and $|\A_j| \leq l + m$.

Second, we consider the case the last operation to construct our treeplex was a Cartesian product.
Let $\Z = \Z_1 \times \cdots \times \Z_k$, and again define $n_i$ as the number of sequences in $\Z_i$, $m_i$ as the number of infosets in $\Z_i$, $l_i$ as the number of leaf sequences in $\Z_i$, and $d-1$ as the maximum depth of any of these subtreeplexes.

Given a standard representation of $\lambda_{\Z_i}(t, \bm{y}_i, \bm{w}_i)$ of size $n_i$ for all $i \in [k]$, by \cref{lemma:projection_cartesian} it takes $O(n \log k)$ to compute a standard representation of  $\lambda_{\Z}(t, \bm{y}, \bm{w})$ of size $n$. By induction, it takes $O\big((d-1) n_i \log (l_i+ m_i) \big)$ to compute $\lambda_{\Z_i}(t, \bm{y}_i, \bm{w}_i)$ for treeplex $\Z_i$. Thus the total computation required to compute $\lambda_{\Z}(t, \bm{y}, \bm{w})$ is 

\begin{align*}
O(n \log k) + \sum_{i \in [k]} O \big((d-1) n_{i} \log  (l_i + m_{i})\big)
&= O(n \log k) + \sum_{i \in [k]} O \big((d-1) n_i \log (l+m) \big) \\
&= O(n \log k) +  O\big((d-1) \sum_{i \in [k]} n_{i} \log (l+m)\big) \\
&= O(n \log k) + O\big((d-1) n \log (l+m)\big) \\
&= O(n \log m) + O\big((d-1) n \log (l+m)\big)\\
&= O(d n \log (l+m))
\end{align*}

since we have necessarily that $l_i \leq l$ and $m_i \leq m$ for all $i \in [k]$, and $k \leq m$. 
\end{proof}

Finally, we are ready to prove the main statement.
\begin{proof}[Proof of \Cref{prop:practical implementation ptbp}]
By \cref{lemma:lambda_fn_lemma}, we know that we can recursively compute a standard representation of $\lambda_{\Z}(t, \bm{y}, \bm{w})$ in $O\big(d n \log (l+m) \big)$ time. Assuming we use this construction, invoking \cref{lemma:projection_branching}, given an optimal value of $t$, we can compute the partial argument corresponding to the values of the sequences that originate at the root infosets, which allow the optimal value to be realized for the value function. Then, we can use optimal arguments for these sequences recursively at the subtreeplexes to continue computing the optimal argument at sequences lower on the treeplex. We can do this because in the process of computing the derivative of the value function of the entire treeplex, we have also computed the derivative of the value function for each of the subtreeplexes. Thus, once we have computed an optimval value of $t$ for the value function at the top level, we can do a top-down pass to compute the optimal values for all sequences that occur at any level in the treeplex. This is detailed in the analysis done in the proof of Lemma 14 in \citet{farina2022near}.

In order to pick the optimal value of $t$ for the value function, since $\lambda_{\Z}(\cdot, \bm{y}, \bm{w})$ is strictly increasing, we only have to consider two cases: $\lambda_{\Z}(0, \bm{y}, \bm{w}) < 0$ and $\lambda_{\Z}(0, \bm{y}, \bm{w}) \geq 0$. In the first case, the value function $\lambda_{\Z}(\cdot, \bm{y}, \bm{w})$ will be minimized when $\lambda_{\Z}(\cdot, \bm{y}, \bm{w})$ is equal to $0$, and this can be directly computed using the standard representation (it will be necessarily $0$ somewhere because it is strictly monotone). In the second case, since $\lambda_{\Z}(\cdot, \bm{y}, \bm{w})$ is strictly monotone and $\lambda_{\Z}(0, \bm{y}, \bm{w}) \geq 0$, we must have that $\lambda_{\Z}(\cdot, \bm{y}, \bm{w}) \geq 0$, which means that $v_{\Z}(\cdot, \bm{y}, \bm{w})$ is minimized at $t^* = 0$.
\end{proof}

\section{Practical Implementation of \smoothptbp{}}\label{app: practical implementation smooth ptbp}
We have the following lemma, which shows that the stable region $\C_{\geq}$ admits a relatively simple formulation.
\begin{lemma}\label{lem:stable-region-representation} 
The stable region \[\C_{\geq}:= \cone(\T) \cap \{ \bm{R} \in \R^{n+1} \; | \; \langle \bm{R},\bm{a}\rangle \geq R_{0}\}\] can be reformulated as follows:
\begin{align*}
\C_{\geq}  & =  \{ \alpha \bm{x} \; | \; \alpha \geq R_{0},\bm{x} \in \T\} \\
& = \{ \bm{x} \in \R^{n+1}_{+} \; | \xnot \geq R_{0}, \sum_{a \in \A_{j}} x_{ja} = x_{p_{j}}, \forall \; j \in \J\}.
\end{align*}
\end{lemma}
\begin{proof}
By definition, we have
\[\C_{\geq} = \{ \bm{R} \in \cone(\T) \; | \; \langle \bm{R},\bm{a}\rangle \geq R_{0}\}.\]
Note that for $\bm{R} \in \cone(\T), \bm{R} = \alpha \bm{x}$ with $\alpha \geq 0$ and $\langle\bm{x},\bm{a}\rangle=1$. Therefore, for $\bm{R} \in \C$ we have $\langle \bm{R},\bm{a}\rangle \geq R_{0} \iff \alpha \geq R_{0}$. This shows that we can write
\[ \C_{\geq} = \{ \alpha \bm{x} \; | \; \alpha \geq R_{0},\bm{x} \in \T\}.\]
Now let $\bm{x} \in \Cgeq$, i.e., let $  \bm{x} = \alpha \hat{\bm{x}}$ with $\alpha \geq R_{0}$ and $ \bm{x} \in \T$. Since $\hat{\bm{x}} \in \T$, we have $\xnot = 1$, so that $\hat{x}_{\varnothing} = \alpha \geq R_{0}$. Additionally, we have $\hat{\bm{x}} \geq 0, \sum_{a \in \A_{j}} \hat{x}_{ja} = \hat{x}_{p_{j}}, \forall \; j \in \J$. Multiplying by $\alpha \geq R_{0}$, we obtain that $\bm{x} \geq 0$ and $\sum_{a \in \A_{j}} x_{ja} = x_{p_{j}}, \forall \; j \in \J$. Overall we have shown
\[ \C_{\geq} \subseteq \{ \bm{x} \in \R^{n+1} \; | \xnot \geq R_{0}, \sum_{a \in \A_{j}} x_{ja} = x_{p_{j}}, \forall \; j \in \J, \bm{x} \geq \bm{0}\}.\]
We now consider $\bm{x} \in \{ \bm{x} \in \R^{n+1} \; | \xnot \geq R_{0}, \sum_{a \in \A_{j}} x_{ja} = x_{p_{j}}, \forall \; j \in \J, \bm{x} \geq \bm{0}\}$ with $\bm{x} \neq \bm{0}$. Then $\bm{x} = \alpha \frac{\bm{x}}{\alpha}$ with $\alpha = \xnot$, so that $\alpha \geq R_{0}$ and 
\[\sum_{a \in \A_{j}} x_{ja} = x_{p_{j}}, \forall \; j \in \J \iff \sum_{a \in \A_{j}} \frac{x_{ja}}{\alpha} = \frac{x_{p_{j}}}{\alpha}, \forall \; j \in \J.\]
Therefore
\[  \{ \bm{x} \in \R^{n+1} \; | \xnot \geq R_{0}, \sum_{a \in \A_{j}} x_{ja} = x_{p_{j}}, \forall \; j \in \J, \bm{x} \geq \bm{0}\} \subseteq \C_{\geq}.\]
This shows that we have
\[ \C_{\geq} = \{ \bm{x} \in \R^{n+1} \; | \xnot \geq R_{0}, \sum_{a \in \A_{j}} x_{ja} = x_{p_{j}}, \forall \; j \in \J, \bm{x} \geq \bm{0}\}.\]
\end{proof}

\begin{proposition}
For a treeplex $\T$ with depth $d$, number of sequences $n$, number of leaf sequences $l$, and number of infosets $m$, the complexity of computing the orthogonal projection of a point $y \in \R^{n+1}$ onto $\C_{\geq}  =  \{ \alpha \bm{x} \; | \; \alpha \geq R_{0},\bm{x} \in \T\}$ is $O\big(d n \log (l+m)\big)$.
\end{proposition}
\begin{proof}
The proof is the same as that for \cref{prop:practical implementation ptbp}, since the derivative of the value function can be computed in $O\big(d n \log (l+m)\big)$ time. However, this time, we have an additional constraint that $t \geq R_0$. Thus instead of checking the sign of $\lambda_{\Z}(\cdot, \bm{y}, \bm{w})$ at $t = 0$, we check the sign at $R_0$.

If $\lambda_{\Z}(R_0, \bm{y}, \bm{w}) < 0$, then because $\lambda_{\Z}(\cdot, \bm{y}, \bm{w})$ is a strictly monotone function, the function will be $0$ for some value of $t$, and this is exactly $t^*$, which minimizes the value function with respect to $t$, when $t \geq R_0$. On the other hand, if $\lambda_{\Z}(R_0, \bm{y}, \bm{w}) \geq 0$, then again because the function is strictly monotone in $t$, we know that the value function must get minimized at $t^* = R_0$. Using the same argument as in the proof of Proposition \ref{prop:practical implementation ptbp}, since we have computed the standard representations of the derivatives of the value functions at all of the treeplexes, we can do a top-down pass to compute the argument which leads to the optimal value of the value function.

\end{proof}
\section{Proof of Theorem \ref{th:regret-guarantees-smooth-sptbp-2playerEFG}}\label{app:proof convergence smooth pbtb NE}
\begin{proof}[Proof of Theorem \ref{th:regret-guarantees-smooth-sptbp-2playerEFG}]
For the sake of conciseness we write $\bm{f}^{x}_{t} = \bm{f}(\bm{x}_{t},\bm{My}_{t})$ and $\bm{f}^{y}_{t} = \bm{f}(\bm{y}_{t},-\bm{M}\tr\bm{x}_{t})$. 

From our Proposition \ref{prop:ptbp-adversarial-regret}, we have that, for the first player,
\[ \sum_{t=1}^{T} \langle \bm{x}_{t} - \hat{\bm{x}}, \bm{My}_{t} \rangle = \sum_{t=1}^{T} \langle \bm{R}_{t}-\hat{\bm{R}},\bm{f}_{t}^{x}\rangle.\]
Now $\sum_{t=1}^{T} \langle \bm{R}_{t}-\hat{\bm{R}},\bm{f}_{t}^{x}\rangle $ is the regret obtained by running Predictive \omd{} on $\Cgeq$ against the sequence of loss $\bm{f}^{x}_{1},...,\bm{f}^{x}_{T}$. From Proposition 5 in \cite{farina2021faster}, we have that
\[  \sum_{t=1}^{T} \langle {\bm{R}}_{t}^{x}-\hat{\bm{R}}^{x},\bm{f}_{t}^{x}\rangle \leq \frac{\| \hat{\bm{R}} \|_{2}^{2}}{2\eta} + \eta \sum_{t=1}^{T} \| \bm{f}^{x}_{t} - \bm{f}^{x}_{t-1}\|_{2}^{2}  - \frac{1}{8\eta} \sum_{t=1}^{T} \| {\bm{R}}^{x}_{t+1} - {\bm{R}}_{t+1}^{x} \|_{2}^{2}.\]
Since $\hat{\bm{R}}_{t} \in \Cgeq$, we can use our Proposition \ref{prop:lipschitness-chd} to show that
\[ \| \bm{x}_{t+1} - \bm{x}_{t} \|_{2}^{2} \leq \frac{\Omega}{R_{0}^{2}} \| {\bm{R}}^{x}_{t+1} - {\bm{R}}_{t+1}^{x} \|_{2}^{2}.\]
This shows that
\[  \sum_{t=1}^{T} \langle {\bm{R}}_{t}^{x}-\hat{\bm{R}}^{x},\bm{f}_{t}^{x}\rangle \leq \frac{\| \hat{\bm{R}} \|_{2}^{2}}{2\eta} + \eta \sum_{t=1}^{T} \| \bm{f}^{x}_{t} - \bm{f}^{x}_{t-1}\|_{2}^{2}  - \frac{R_{0}^{2}}{ 8 \Omega^2  \eta} \sum_{t=1}^{T} \| {\bm{R}}^{x}_{t+1} - {\bm{R}}_{t+1}^{x} \|_{2}^{2}\]
which gives, using the norm equivalence $\| \cdot \|_{2} \leq \| \cdot \|_{1} \leq \sqrt{n+1}\| \cdot \|_{2}$, the following inequality:
\[  \sum_{t=1}^{T} \langle {\bm{R}}_{t}^{x}-\hat{\bm{R}}^{x},\bm{f}_{t}^{x}\rangle \leq \frac{\| \hat{\bm{R}} \|_{2}^{2}}{2\eta} + \eta \sum_{t=1}^{T} \| \bm{f}^{x}_{t} - \bm{f}^{x}_{t-1}\|_{1}^{2}  - \frac{R_{0}^{2}}{ 8 \Omega^2 (n+1) \eta}\sum_{t=1}^{T}  \| {\bm{R}}^{x}_{t+1} - {\bm{R}}_{t+1}^{x} \|_{2}^{2}\]
The above inequality is a RVU bound:
\[  \sum_{t=1}^{T} \langle {\bm{R}}_{t}^{x}-\hat{\bm{R}}^{x},\bm{f}_{t}^{x}\rangle \leq \alpha + \beta \sum_{t=1}^{T} \| \bm{f}^{x}_{t} - \bm{f}^{x}_{t-1}\|_{1}^{2}  - \gamma \sum_{t=1}^{T}\| {\bm{R}}^{x}_{t+1} - {\bm{R}}_{t+1}^{x} \|_{2}^{2}\]
with
\begin{equation}\label{eq:coefficient-RVU}
    \alpha = \frac{\| \hat{\bm{R}} \|_{2}^{2}}{2\eta}, \beta = \eta, \gamma =  \frac{R_{0}^{2}}{ 8 \Omega^2 (n+1) \eta}.
\end{equation}
To invoke Theorem 4 in \cite{syrgkanis2015fast}, we also need the utilities of each player to be bounded by $1$. This can be done can rescaling $\bm{f}^{x}_{t}=\bm{My}_{t}$ and $\bm{f}^{y}_{t} = -\bm{M}\tr\bm{x}_{t}$. In particular, we know that
\[ \| \bm{My}\|_{\infty} \leq \| \bm{M}\|_{\ell_{2},\ell_{\infty}} \| \bm{y} \|_{2} \leq \| \bm{M}\|_{\ell_{2},\ell_{\infty}} \cdot \hat{\Omega}\]
with $\| \bm{M}\|_{\ell_{2},\ell_{\infty}} = \max_{i \in [n+1]} \| \left(A_{ij}\right)_{j \in [m+1]}\|_{2}$ and $\hat{\Omega} = \max \{ \max \{ \| \bm{x}\|_{2},\| \bm{y}\|_{2}\} \; \bm{x} \in \cX,\bm{y} \in \cY\}$. This corresponds to multiplying $\beta$ in \eqref{eq:coefficient-RVU} by $\| \bm{M} \| \times \hat{\Omega}$ with $\| \bm{M} \|:= \max \{ \| \bm{M}\|_{\ell_{2},\ell_{\infty}}, \| \bm{M}\tr\|_{\ell_{2},\ell_{\infty}} \}$.
To apply Theorem 4 in \cite{syrgkanis2015fast} we also need $\beta \leq \gamma$.  Since we need the same condition for the second player, we take 
\[\eta = R_{0}\left(\sqrt{8d\hat{\Omega}^3}\| \bm{M}\|\right)^{-1}.\]
Under this condition on the stepsize, we can invoke Theorem 4 in \cite{syrgkanis2015fast} to conclude that 
\[ \sum_{t=1}^{T} \langle {\bm{R}}_{t}^{x}-\hat{\bm{R}}^{x},\bm{f}_{t}^{x}\rangle + \sum_{t=1}^{T} \langle {\bm{R}}_{t}^{y}-\hat{\bm{R}}^{y},\bm{f}_{t}^{y}\rangle \leq \frac{\| \hat{\bm{R}}^{x} \|_{2}^{2} + \| \hat{\bm{R}}^{y} \|_{2}^{2}}{\eta}.\]
Since the duality gap is bounded by the average of the sum of the regrets of both players~\cite{freund1999adaptive}, and replacing $\eta$ by its expression, we obtain that 
\[ \max_{\bm{y} \in \cY} \; \langle \bar{\bm{x}}_{T},\bm{My}\rangle - \min_{\bm{x} \in \cX} \; \langle \bm{x},\bm{M}\bar{\bm{y}}_{T}\rangle \leq \frac{2 \hat{\Omega}^2}{\eta}\frac{1}{T}.\]
\end{proof}

\section{Details on Numerical Experiments}\label{app:simu}
\subsection{Additional Algorithms}
\paragraph{\adamtbp.} We present \adamtbp, our instantiation of Algorithm \ref{alg:blackwell-approachability based regmin} inspired from the adaptive algorithm \adam~\cite{kingma2014adam} in Algorithm \ref{alg:adam-stable-blackwell-treeplex}. Since \adam{} is not necessarily a regret minimizer~\citep{reddi2019convergence}, there are no regret guarantees for \adamtbp. We choose to consider this algorithm for the sake of completeness, since \adam{} is widely used in other settings.
\begin{algorithm}
      \caption{Adam\tbp}
      \label{alg:adam-stable-blackwell-treeplex}
      \begin{algorithmic}[1]
     \State {\bf Input}: $\eta, \delta>0, \beta_{1}, \beta_{2} \in [0,1]$
      \State {\bf Initialization}: $\bm{R}_{1} = \bm{0} \in \R^{n+1},\bm{s}_{0}=\bm{0} \in \R^{n+1},\bm{g}_{0} = \bm{0} \in \R^{n+1}$
      \For{$t = 1, \dots, T$}
      \State $\bm{x}_{t} = \bm{R}_{t}/\langle\bm{R}_{t},\bm{a}\rangle$
      \State Observe the loss vector $\bm{\ell}_{t} \in \R^{n+1}$
      \State $\bm{s}_{t} = \beta_{2} \bm{s}_{t-1} + (1-\beta_{2}) \bm{f}(\bm{x}_{t},\bm{\ell}_{t}) \odot \bm{f}(\bm{x}_{t},\bm{\ell}_{t}) $
      \State $\hat{\bm{s}}_{t} = \bm{s}_{t}/(1-\beta^{t}_{2})$
      \State $\bm{g}_{t} = \beta_{1} \bm{g}_{t-1} + (1-\beta_{1})\bm{f}(\bm{x}_{t},\bm{\ell}_{t})$
      \State $\hat{\bm{g}}_{t} = \bm{g}_{t}/(1-\beta^{t}_{1})$
      \State $\bm{H}_{t} = \diag\left(\sqrt{\hat{\bm{s}}_{t}}+\epsilon \bm{1}\right)$
      \State $\bm{R}_{t+1} \in \Pi_{\C}^{\bm{H}_{t}} \left(\bm{R}_{t} - \eta \bm{H}^{-1}_{t} \hat{\bm{g}}_{t} \right)$
      \EndFor
\end{algorithmic}
\end{algorithm}
\paragraph{Single-call Predictive Online Mirror Descent (\scpomd).}

We present \scpomd{} in Algorithm \ref{alg:scpomd}. This algorithm runs a variant of predictive online mirror descent with only one orthogonal projection at every iteration~\citep{joulani2017modular}. The pseudocode from Algorithm \ref{alg:scpomd} corresponds to choosing the squared $\ell_{2}$-norm as a distance generating function - in principle, other distance generating functions are possible, e.g. dilated entropy~\cite{farina2021better}. Combined with the self-play framework, \scpomd{} ensures that the average of the visited iterates converges to a Nash equilibrium at a rate of $O(1/T)$, similar to the variant of predictive online mirror descent with two orthogonal projections at every iteration~\citep{farina2021better}.
\begin{algorithm}[H]
      \caption{Single-call predictive online mirror descent (\scpomd)}
      \label{alg:scpomd}
      \begin{algorithmic}[1]
     \State {\bf Input}: $\eta >0$,

      \State {\bf Initialization}: $\bm{x}_{0} = \bm{\ell}_{0} = \bm{\ell}_{-1} = \bm{0} \in \R^{n+1}$
      
      \For{$t = 1, \dots, T$}
      \State $\bm{x}_{t} = \Pi_{\T} \left( \bm{x}_{t-1} - \eta \left(2\bm{\ell}_{t-1} - \bm{\ell}_{t-2}\right) \right) $
      \State Observe the loss vector $\bm{\ell}_{t} \in \R^{n+1}$
      \EndFor
\end{algorithmic}
\end{algorithm}
\subsection{Algorithm Implementation Details}
All algorithms are initialized using the uniform strategy (placing equal probability on each action at each decision point). For algorithms that are not stepsize invariant (\smoothptbp{} and \scpomd{}), we try stepsizes in $\eta \in \{0.05, 0.1, 0.5, 1, 2, 5\}$ and we present the performance with the best stepsize. For \smoothptbp{}, we use $R_0 = 0.1$. 
For both \adatbp{} and \adamtbp{}, we use $\delta= 1 \times 10^{-6}$, and for \adamtbp{} we use $\beta_1 = 0.9$ and $\beta_2 = 0.999$.

\subsection{Comparing the Performance of our Algorithms}\label{fig:comparison tb algo}
In Figure \ref{fig:tb_algo_comparison}
we compare the performance of \tbp, \ptbp, \smoothptbp, \adatbp{} and \adamtbp.

It can be seen that \ptbp{} and \smoothptbp{} perform similarly, both when using quadratic averaging and when using the last iterate, and they generally outperform the other algorithms. In Kuhn, Liar's Dice, and Battleship, the last iterate seems to perform quite well, whereas in Leduc and Goofspiel, the quadratic averaging scheme works better. \adamtbp{} seems to not converge in any of the games, which is not surprising, because it does not have theoretical guarantees for convergence.

\begin{figure*}[htb]
\centering
\includegraphics[width=\textwidth]{"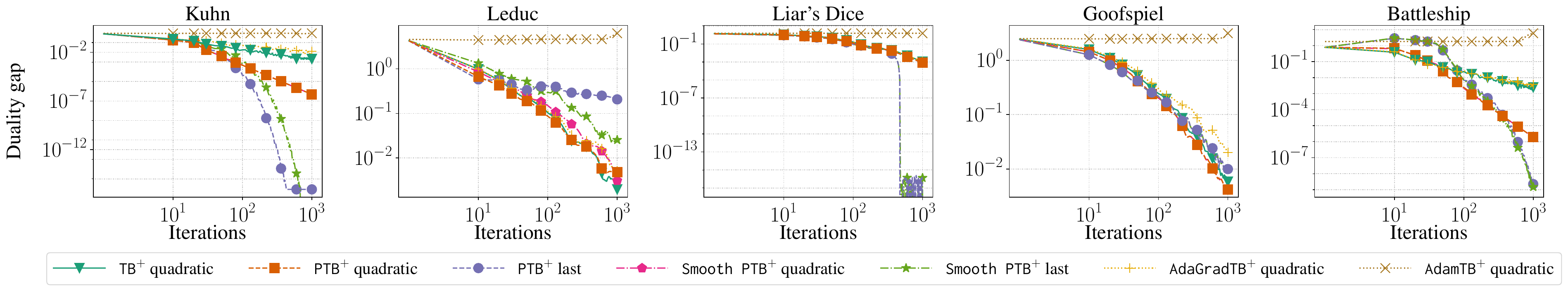"}
\caption{Convergence to Nash equilibrium as a function of number of iterations for \tbp{} with quadratic averaging, \ptbp{} with quadratic averaging and last iterate, and \smoothptbp{} with quadratic averaging and last iterate. Every algorithm is using alternation.}
\label{fig:tb_algo_comparison}
\end{figure*}
\subsection{Individual Performance}\label{app:individual performance}
In \crefrange{fig:tbp}{fig:scpomd}, we compare the individual performance of \tbp, \ptbp, \smoothptbp, \adatbp, \adamtbp, \cfrp, \pcfrp{} and \scpomd{} with different weighting schemes, with and without alternation. We also show the performance of the last iterate. The goal is to choose the most favorable framework for each algorithms, in order to have a fair comparison. We find that all algorithms benefit from using alternation. \cfrp{} enjoys stronger performance using linear weights, whereas \ptbp, \pcfrp{} and \scpomd{} have stronger performances with quadratic weights. For this reason this is the setup that we present for comparing the performance of these algorithms in our main body (Figure \ref{fig:all_algo_avg_comparison}).
\begin{figure*}[htp]
\centering
\includegraphics[width=\textwidth]{"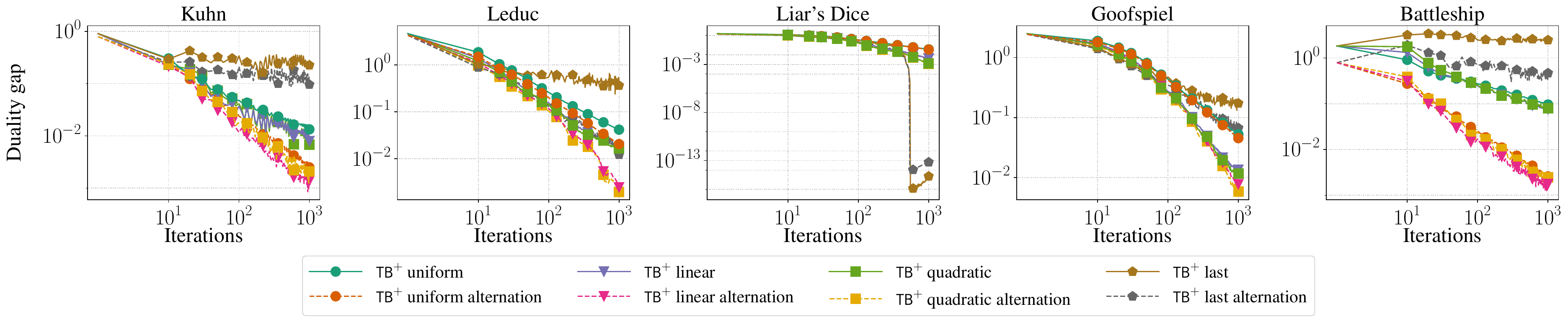"}
\caption{Convergence to Nash equilibrium as a function of number of iterations using uniform, linear, and quadratic averaging, as well as the last iterate, with and without alternation for \tbp{}.}
\label{fig:tbp}
\end{figure*}

\begin{figure*}[htp]
\centering
\includegraphics[width=\textwidth]{"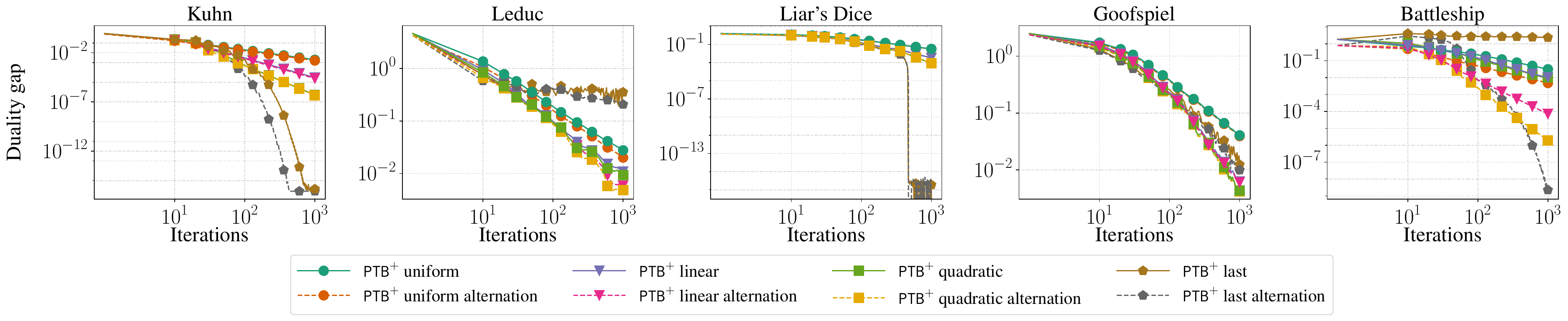"}
\caption{Convergence to Nash equilibrium as a function of number of iterations using uniform, linear, and quadratic averaging, as well as the last iterate, with and without alternation for \ptbp{}.}
\label{fig:ptbp}
\end{figure*}

\begin{figure*}[htp]
\centering
\includegraphics[width=\textwidth]{"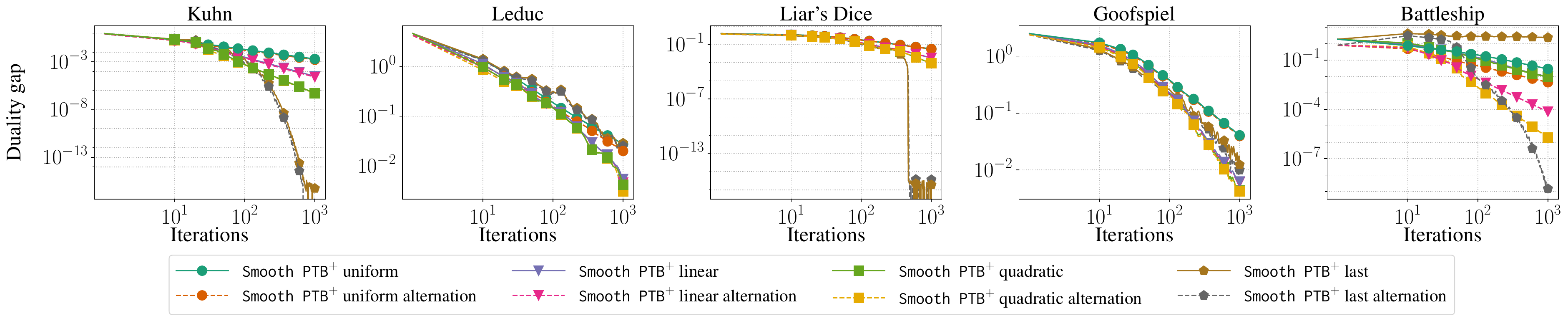"}
\caption{Convergence to Nash equilibrium as a function of number of iterations using uniform, linear, and quadratic averaging, as well as the last iterate, with and without alternation for \smoothptbp{}.}
\label{fig:smoothptbp}
\end{figure*}

\begin{figure*}[htp]
\centering
\includegraphics[width=\textwidth]{"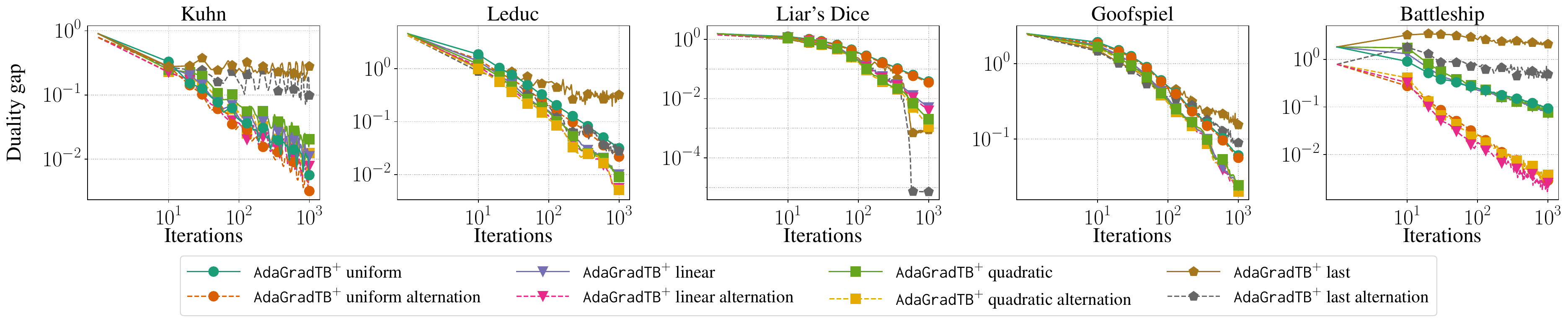"}
\caption{Convergence to Nash equilibrium as a function of number of iterations using uniform, linear, and quadratic averaging, as well as the last iterate, with and without alternation for \adatbp{}.}
\label{fig:adatbp}
\end{figure*}

\begin{figure*}[htp]
\centering
\includegraphics[width=\textwidth]{"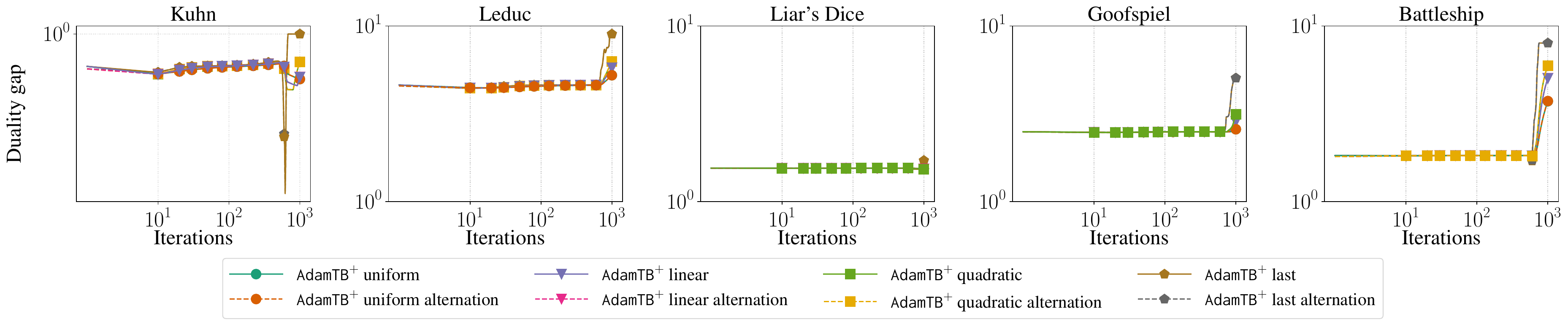"}
\caption{Convergence to Nash equilibrium as a function of number of iterations using uniform, linear, and quadratic averaging, as well as the last iterate, with and without alternation for \adatbp{}.}
\label{fig:adamtbp}
\end{figure*}

\begin{figure*}[htp]
\centering
\includegraphics[width=\textwidth]{"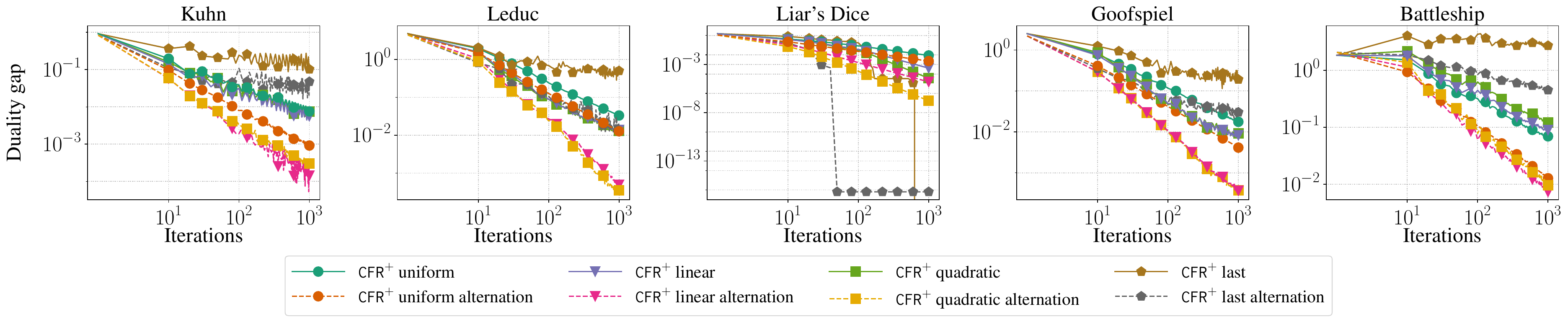"}
\caption{Convergence to Nash equilibrium as a function of number of iterations using uniform, linear, and quadratic averaging, as well as the last iterate, with and without alternation for \cfrp{}.}
\label{fig:cfrp}
\end{figure*}

\begin{figure*}[htp]
\centering
\includegraphics[width=\textwidth]{"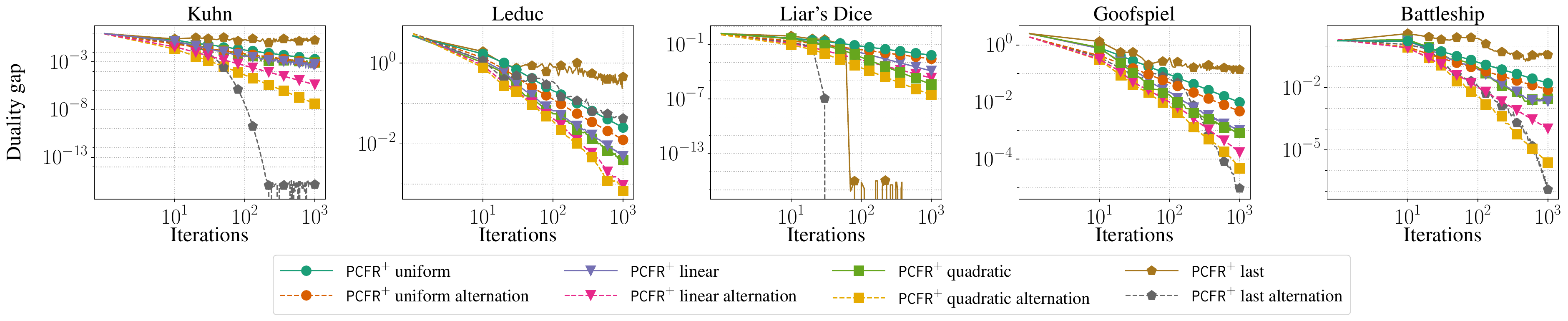"}
\caption{Convergence to Nash equilibrium as a function of number of iterations using uniform, linear, and quadratic averaging, as well as the last iterate, with and without alternation for \pcfrp{}.}
\label{fig:pcfrp}
\end{figure*}

\begin{figure*}[htp]
\centering
\includegraphics[width=\textwidth]{"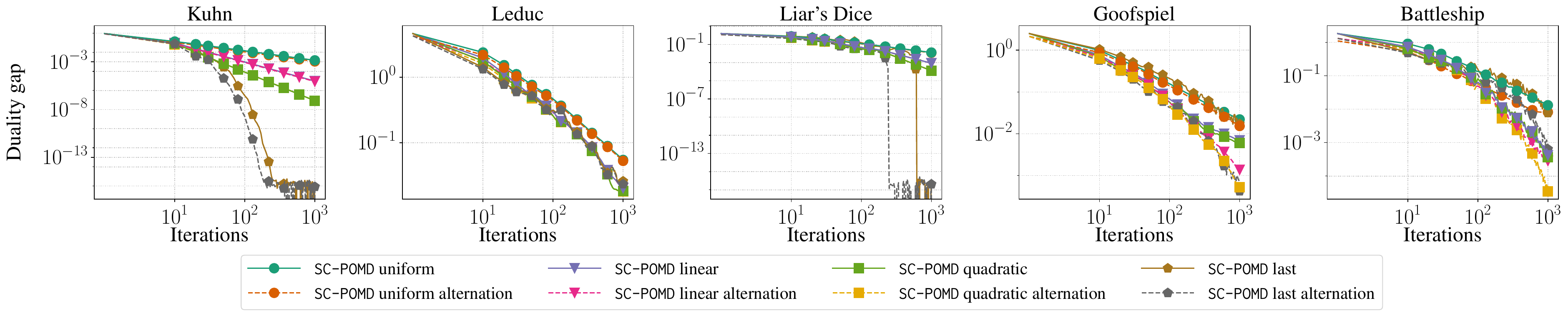"}
\caption{Convergence to Nash equilibrium as a function of number of iterations using uniform, linear, and quadratic averaging, as well as the last iterate, with and without alternation for  \scpomd{}.}
\label{fig:scpomd}
\end{figure*}

\end{document}